\documentclass{llncs}
\usepackage{mathtools, xparse}
\usepackage{graphicx}
\usepackage{caption}
\usepackage{subcaption}
\usepackage{epstopdf}
\usepackage{bm}
\usepackage{url}
\usepackage[colorlinks = true,
            linkcolor = blue,
            urlcolor  = blue,
            citecolor = blue,
            anchorcolor = blue]{hyperref}
\usepackage{scalerel}
\usepackage{amsfonts,amsmath}
\usepackage{algorithm,algorithmic}
\usepackage{mathtools}
\usepackage{multirow}
\usepackage{verbatim}
\usepackage{url}
\usepackage{comment}
\usepackage{appendix}
\usepackage{mathtools}
\usepackage{epsf,pgf,graphicx}
\usepackage{color}
\usepackage{tikz,pgf}
\usetikzlibrary{arrows.meta}
\usepackage{diagbox}
\usepackage{makecell}
\usepackage{color}
\usepackage{braket}
\usepackage{mathtools}
\usepackage{float}
\usepackage{tikz}
\usetikzlibrary {positioning}
\usepackage{pgf}
\usepackage{qcircuit}
\usepackage{multirow}
\usepackage{multicol}
\usepackage{listings}
\usepackage{array}
\usepackage{makecell}
\usepackage[utf8]{inputenc}
\usepackage[english]{babel}
\usepackage[switch]{lineno}

\newtheorem{cons}{Construction}
\DeclareFixedFont{\ttb}{T1}{txtt}{bx}{n}{7.5} 
\DeclareFixedFont{\ttm}{T1}{txtt}{m}{n}{7.5}  

\usepackage{color}
\definecolor{deepblue}{rgb}{0,0,0.5}
\definecolor{deepred}{rgb}{0.6,0,0}
\definecolor{deepgreen}{rgb}{0,0.5,0}

\usepackage{listings}

\newcommand\pythonstyle{\lstset{
language=Python,
basicstyle=\ttm,
morekeywords={self},              
keywordstyle=\ttb\color{deepblue},
emph={MyClass,__init__},          
emphstyle=\ttb\color{deepred},    
stringstyle=\color{deepgreen},
frame=tb,                         
showstringspaces=false
}}

\lstnewenvironment{python}[1][]
{
\pythonstyle
\lstset{#1}
}
{}

\newcommand\pythoninline[1]{{\pythonstyle\lstinline!#1!}}

\begin{document}
	\title{Further Constructions of AMUBs for Non-prime power Composite Dimensions}
	\author{Ajeet Kumar and Subhamoy Maitra}
	\institute{Applied Statistics Unit, Indian Statistical Institute Kolkata\\
	\email{ajeetk52@gmail.com,
	subho@isical.ac.in}}
	\date{}
	\maketitle

\begin{abstract}
Construction of a large class of Mutually Unbiased Bases (MUBs) for non-prime power composite dimensions ($d = k\times s$) is a long standing open problem, which leads to different construction methods for the class Approximate MUBs (AMUBs) by relaxing the criterion that the absolute value of the dot product between two vectors chosen from different bases should be $\leq \frac{\beta}{\sqrt{d}}$. In this chapter, we consider a more general class of AMUBs (ARMUBs, considering the real ones too), compared to our earlier work in [Cryptography and Communications, 14(3): 527--549, 2022]. We note that the quality of AMUBs (ARMUBs) constructed using RBD$(X,A)$ with $|X|= d$, critically depends on the parameters, $|s-k|$, $\mu$ (maximum number of elements common between any pair of blocks), and the set of block sizes. We present the construction of $\mathcal{O}(\sqrt{d})$ many $\beta$-AMUBs for composite $d$ when $|s-k|< \sqrt{d}$, using RBDs having block sizes approximately $\sqrt{d}$, such that $|\braket{\psi^l_i|\psi^m_j}| \leq \frac{\beta}{\sqrt{d}}$ where $\beta = 1 + \frac{|s-k|}{2\sqrt{d}}+ \mathcal{O}(d^{-1}) \leq 2$. Moreover, if real Hadamard matrix of order $k$ or $s$ exists, then one can construct at least $N(k)+1$ (or $N(s)+1$) many $\beta$-ARMUBs for dimension $d$, with $\beta \leq 2 - \frac{|s-k|}{2\sqrt{d}}+ \mathcal{O}(d^{-1})< 2$, where $N(w)$ is the number of MOLS$(w)$. This improves and generalizes some of our previous results for ARMUBs from two points, viz., the real cases are now extended to complex ones too. The earlier efforts use some existing RBDs, whereas here we consider new instances of RBDs that provide better results. Similar to the earlier cases, the AMUBs (ARMUBs) constructed using RBDs are in general very sparse, where the sparsity $(\epsilon)$ is $1 - \mathcal{O}(d^{-\frac{1}{2}})$. 
\end{abstract}

\section{Introduction}
\label{sec:intro}
Mutually unbiased bases are structures over Hilbert spaces with various applications in several research areas like quantum tomography, state estimation, quantum key distribution, dense coding, entanglement swapping, etc. When MUBs are constructed over $\mathbb{R}^d$, we get Real MUBs.  The MUBs over $\mathbb{R}^d$ , have interesting connections with Quadratic Forms \cite{cameron1991quadratic}, Association Schemes \cite{ReMUB-AssociationSchemes,AssociationScheme-Coding}, Equi-angular Lines ,  Equi-angular Tight Frames, Fusion Frames over $\mathbb{R}^d$ \cite{ReMUB-FusionFrames}, Mutually Unbiased Real Hadamard Matrices, Bi-angular vectors over $\mathbb{R}^d$ \cite{Holzmann2010,Best2013,Best2015,Kharaghani2018} and applications in construction of Codes \cite{Calderbank1997}.  Various authors have explored connections between MUBs and geometrical objects such as {polytopes} and {projective planes} \cite{MUB-Polytopes,MUB-CompPolytopes,saniga2004mutually,MUBs-HjelmslevGeometry,GaloisUnitaries-MUB}.

There can be maximum $d+1$ MUBs in $\mathbb{C}^d$ and $d/2+1$ MUBs in $\mathbb{R}^d$, and when they exist for a  dimension $d$, we call it a complete set of MUBs. Over $\mathbb{R}^d$ MUBs are known to be non-existent when $d$ is odd, and for most of the other even $d$, there are mostly 3 Real MUBs. The known methods for constructing MUBs over $\mathbb{C}^d$ provide complete sets only when $d$ is some power of prime. In fact if $d = p_1^{n_1}p_2^{n_2}\ldots p_m^{n_m}$ then, the lower bound on number of MUBs is $p_r^{n_r}+1$ where $p_r^{n_r}$ is $\min\{{p_1^{n_1}, p_2^{n_2}, \ldots,p_m^{n_m}}\}$. Such bases are even smaller when considering the problem over the real vector space $\mathbb{R}^d$. A large number of real MUBs are non-existent for most of the dimensions \cite{MUB2}. In fact, only for $d = 4^s, s>1$, we have $d/2 + 1$ many MUBs, whereas for most of the dimensions $d$, which are not a perfect square, we have at best only $2$ Real MUBs \cite{MUB2}.

Given this, various attempts have been made to construct a large class of approximate MUBs (AMUBs) in every dimension. To emphasize that the construction is over $\mathbb{R}^d$, we call it approximate real MUBs (ARMUBs), which are available in literature \cite{ak21,ak22}. The Approximate MUBs have been defined in various manners by different authors. The cue has been taken from the two initial papers \cite{AMUB1,AMUB2} on this. Although the first one focuses on constructing an approximate SIC POVM, the definition of `approximate' has been carried over to the MUB case as well. Various mathematical meaning of Approximation which has been used for relaxing the condition on the absolute value of the dot product between two vectors say $\ket{u}, \ket{v}$, in these two papers,  are $ |\braket{\psi^l_i|\psi^m_j}| \leq \frac{1+o(1)}{\sqrt{d}}$,  $ |\braket{\psi^l_i|\psi^m_j}| \leq \frac{2+o(1)}{\sqrt{d}}$, $ |\braket{\psi^l_i|\psi^m_j}| = \mathcal{O}(\frac{log(d)}{\sqrt{d}})$,  $ |\braket{\psi^l_i|\psi^m_j}| = \mathcal{O}(\frac{1}{\sqrt[4]{d}})$, and $ |\braket{\psi^l_i|\psi^m_j}| = \mathcal{O}(\frac{1}{\sqrt{d}})$ and subsequent researchers investigating approximate MUBs have adopted them as the mathematical definition of `approximate' \cite{AMUB-MixedCharacterSum,AMUB-CharacterSum,AMUB-FrobeniusRing,AMUB-GaloisRing,ARMUB-fromComplexAMUB}.

In this direction, we have defined $\beta-$AMUBs and APMUBs. The former being AMUBs such that $ |\braket{\psi^l_i|\psi^m_j}| \leq \frac{\beta}{\sqrt{d}}$ where $\beta $ is some small real constant for all $d$ and later defined APMUBs is $\beta-$AMUB, such that $\beta = 1+ \mathcal{O}(d^{-\lambda}) \leq 2$ for $d$ where some $\lambda > 0$ is constant and the absolute value of the dot product vectors from different basis can have only one non zero value i.e  $ |\braket{\psi^l_i|\psi^m_j}| \in \{0, \frac{\beta}{\sqrt{d}} \}$. Thus, APMUBs are the subset of $\beta-$AMUBs. We will see that AMUBs constructed using RBD are generally $\beta -$ AMUBs, and basis vectors are very sparse. We have used Big $\mathcal{O}$ Notation to mean that if $f(x) = \mathcal{O}(g(x))$ then there exist a constant $c > 0$, and $x_o$, such that $f(x)\leq c g(x)\, \, \forall x > x_o$.


\subsection{Organization and contribution}
In Section \ref{sec:pre}, we begin by formally defining Mutually Unbiased Bases (MUBs) and approximate MUBs. We also define the important category of AMUBs, which we refer to as $\beta$-AMUBs, and then briefly outline the results related to AMUBs in the literature. Additionally, we provide general characteristics of these known constructions of AMUBs. Next, we present a brief background on the combinatorial design known as the Resolvable Block Design (RBD), which plays a crucial role in our construction of AMUBs.

In Section \ref{sec:TheoreticalAnalysis}, we provide a theoretical analysis of AMUBs that could be constructed using RBD and describe the important parameters of RBD that affect the quality of the constructed AMUBs using RBDs. For this, we categorize RBD into two categories: one having variable block size and another having constant block size. We first show that the sparsity of the AMUBs constructed using RBD is approximately $1- \frac{1}{\sqrt{d}}$ if block sizes are around $\sqrt{d}$. Then, we present a general theorem on AMUBs, assuming the existence of a certain kind of RBD. The block sizes and the maximum number of elements common between any pair of blocks in the RBD play a crucial role in the value of $\beta$. We demonstrate that when block sizes are around $\sqrt{d}$, we obtain very sparse $\beta$-AMUBs with $\beta = \mu +\mathcal{O}(d^{-\frac{1}{2}})$ and sparsity $\epsilon \sim 1- \frac{1}{\sqrt{d}}$, thus showing that as $d$ increases, $\beta$ approaches $\mu$ and the sparsity approaches $1$. For RBDs $(X, A)$, with $|X|= d=k\times s$ having constant block size either $k$ or $s$, we show that $\delta = \frac{|s-k|}{2}$ plays a crucial role in deciding the quality of AMUBs apart from $\mu$. For constant block sizes, we express $d=(q-e)(q\pm f)$ where $q$ is some power of a prime. After that, we provide an estimate of $e$ and $f$ using unconditional results on the gaps between primes and Cram\'er's conjecture.

Section \ref{sec:4} discusses algorithms for constructing the  RBDs for composite $d$. And then we give result about $\beta$-AMUB, which can be constructed using such RBD. In the first subsection, we give the construction for variable block sizes of RBD, but $\mu =1$, and the number of parallel classes is greater than $\sqrt{d}$. In the following subsection, we provide construction using the RBD with constant block sizes, where $\mu$ is either 1 or 2, and the number of parallel classes is $N(s)$ or $N(k)$. We show that RBDs having constant block size can be used to construct $\beta$-ARMUBs if a Hadamard matrix of order $s$ or $k$ is available. We also illustrate our constructions with examples and show how these constructions improve and generalize the previous results.

In Section \ref{sec:5}, we discuss and compare the present results with existing results. In Section \ref{sec:con}, we summarize the main ideas of this work and conclude by suggesting further research possibilities in this direction.

\section{Preliminaries}
\label{sec:pre}
Let $M_l = \left\{\ket{\psi_1^l}, \ldots, \ket{\psi_{d}^l}\right\}$ and $M_m = \left\{\ket{\psi_1^m}, \ldots, \ket{\psi_{d}^m}\right\}$ be two orthonormal bases in the $d$-dimensional vector space over $\mathbb{C}^d$ (or $\mathbb{R}^d$). These two bases will be called Mutually Unbiased if
\begin{equation}
\left|\braket{\psi_i^l | \psi_j^m}\right| = \displaystyle\frac{1}{\sqrt{d}}, \ \forall i, j \in \left \{1, 2, \ldots, d\right\}.  
\end{equation}
 
Now, let us introduce some notations from \cite{ak22}. 
A set of orthonormal bases would be denoted as $\mathbb{M} = \left\{M_1, M_2, \ldots, M_r\right\}$ (which may not necessarily be mutually unbiased bases) of dimension $d$. Here, $\Delta$ denotes the set of different inner product values between the vectors from different orthonormal bases. In other words, $\Delta$ contains the distinct values of $\left|\braket{\psi_i^l | \psi_j^m}\right|$ for all $i,j\in \left\{1, 2, \ldots, d\right\}$ and $l \neq m \in \left\{1, \ldots, r\right\}$. The set $\mathbb{M} = \left\{M_1, M_2, \ldots, M_r\right\}$ consisting of such orthonormal bases will form a set of mutually unbiased bases (MUBs) of size $r$ provided that $M_l, M_m \in \mathbb{M}$ are mutually unbiased for all $l \neq m$. It is worth noting that when $\mathbb{M}$ forms an MUB, $\Delta$ is a singleton set with the only element $\frac{1}{\sqrt{d}}$. However, there will be more than one value in the set for approximate ones. The notation $\beta$-ARMUB (Approximate Real MUB) has been used \cite{ak22} in the context of a set of orthonormal bases $\mathbb{M}$, where the maximum value in $\Delta$ is bounded by $\frac{\beta}{\sqrt{d}}$, with $\beta$ being some small real constant $\geq 1$. It is worth noting that this definition is significant for situations where $d$ can be increased and $\beta$ remains bounded by some constant. We will refer to such sets of MUBs as $\beta$-AMUB (Approximate MUB) for both complex and real cases. If we need to emphasize something specific in the context of Approximate Real MUBs, we will refer to it as $\beta$-ARMUB.

The sparsity of orthonormal bases is an essential feature of our construction. To characterize the sparsity of the MUBs/AMUBs/ARMUBs, we arrange the basis vectors as columns of a $d \times d$ matrix and use the standard measure of sparsity (denoted by $\epsilon$) as "the number of zero elements in the matrix divided by the total number of elements." The columns of the matrix consist of orthonormal basis vectors. The closer the value of $\epsilon$ is to unity, the more the number of zeros in the matrix. It will be shown that for our construction, the sparsity generally varies as $\epsilon = 1 - \mathcal{O}(d^{-\frac{1}{2}})$.

Klappenecker et al. \cite{AMUB1} for the first time, they introduced the notion of AMUB. They showed that \cite[Theorem 11]{AMUB1}, for all $d$, one can construct $d$ many bases such that
$$|\braket{\psi^l_i|\psi^m_j}| = \mathcal{O}(d^{-\frac{1}{3} }) \Rightarrow \beta = \mathcal{O}(d^{\frac{1}{6} })$$ 
and can construct $d^m, m \geq 2$ many bases such that 
$$|\braket{\psi^l_i|\psi^m_j}| = \mathcal{O}(d^{-\frac{1}{4}}) \Rightarrow \beta =  \mathcal{O}(d^{\frac{1}{4} }),$$
where $\ket{\psi_i^l}$ and $\ket{\psi_j^m}$ are basis vectors from different bases. Thus, note that $\beta$ becomes large if AMUB increases. We will see that this feature will be there in every construction, though the functional dependence would be very different. The results for all $d$ have been improved by \cite[Theorem 1]{AMUB2} where the construction based on the finite field showed that for all $d$, there are $d$ many AMUB such that
$$|\braket{\psi^l_i|\psi^m_j}| \leq \left(\frac{2}{\sqrt{\pi}}+ \mathcal{O}\left(\log^{-1} n\right)\right)\left(\frac{\log d}{ d } \right)^{\frac{1}{2}} \Rightarrow \beta =\mathcal{O}( \sqrt{ \log d }~),$$
which was further improved in construction based on the elliptic curve, \cite[Theorem 2]{AMUB2} where the construction gave $p^{m-1}$, $m\geq 2$ where $p$ is a prime such that $\sqrt{n} -1 \leq \sqrt{p} \leq \sqrt{n} +1.$
$$|\braket{\psi^l_i|\psi^m_j}| \leq \frac{2m + \mathcal{O}(d^{-\frac{1}{2} })}{\sqrt{d}} = \mathcal{O}(n^{-\frac{1}{2}}) \Rightarrow \beta = 2m + \mathcal{O}(d^{-\frac{1}{2} }).$$
The above are the only known constructions applicable for all $d$. The minimum value of $m= 2$, hence $\beta = 4 + \mathcal{O}(d^{-\frac{1}{2} })$ which would be the minimum possible value of $\beta$. For this case, we get $d$ AMUBs. Furthermore, there is no possibility of Real AMUBs using this construction as all the components of the basis vectors are products of the character of group elements and the complex numbers on unit modulus.

In pursuit of better constructions of AMUB, various authors have given different results, however, they are not generic in nature, rather for some specific kind of dimension. Some of them we summaries below. The following AMUB was defined strictly when $|\braket{\psi^l_i|\psi^m_j}|\leq \frac{1+o(1)}{\sqrt{d}}$ .

In \cite{AMUB1}, it was shown that for the dimensions of the form $d= p-1$, where $p$ is some prime, there exists $d+1$ bases such that
$$|\braket{\psi^l_i|\psi^m_j}| \leq \frac{1}{\sqrt{d}} + \mathcal{O}(d^{-1}) \Rightarrow \beta = 1+ \mathcal{O}(d^{-\frac{1}{2}}),$$
where $\ket{\psi_i^l}$ and $\ket{\psi_j^m}$ are vectors from different bases. In \cite[Theorem 3.1]{AMUB-GaloisRing}, a construction of AMUB has been given using Galois ring. The authors have constructed $q+1= \mathcal{O}(\sqrt{d})$ many AMUB for dimension $d= q(q-1)$ where $q$ is some prime-power, satisfying
$$|\braket{\psi^l_i|\psi^m_j}| \leq \frac{1}{q-1} \Rightarrow \beta = 1+ \mathcal{O}(d^{-\frac{1}{2}}).$$ 
In \cite{AMUB-CharacterSum} AMUB has been constructed using orthogonality of character sum over finite field for some prime power dimensions, i.e., $d=p^m$ over $\mathbb{C}^d$, specifically when $p=2$. However, we know that for such $d$, there exist well-known construction for the compete set MUBs. Thus, their results do not appears to be of any special interest.

In \cite{AMUB-MixedCharacterSum} the authors have given construction for $d+1$ or $d+2$ many AMUB over $\mathbb{C}^d$ when  $d= q-1$ where $q$ is some power of prime. Here, it is shown that one can get more than $d+1$ many AMUB with $\beta = 1+ \mathcal{O}(d^{-\frac{1}{2}})$, whereas the maximum possible MUBs for any $d$ is $d+1$. The method employed mixed character sum of certain special functions over the finite field. \cite[Theorem 3.2]{AMUB-MixedCharacterSum} constructed $q$ many AMUB when $d= q-1$ where 
$$|\braket{\psi^l_i|\psi^m_j}| = \frac{1+\sqrt{d}}{d}  ~~\text{or} ~~ 0\Rightarrow \beta = 1+ \mathcal{O}(d^{-\frac{1}{2}}).$$ 
Note that there is an equality in the above relationship for the value of $|\braket{\psi^l_i|\psi^m_j}|$ when the vectors are from different bases. Since $\Delta = {0, \frac{1+\sqrt{d}}{d} }$ consists of just two values, satisfying the condition of APMUB. The author similarly showed in \cite[Theorem 3.5]{AMUB-MixedCharacterSum} that $q+1$ AMUBs exist when $d= q-1$, where
$$|\braket{\psi^l_i|\psi^m_j}| = \frac{\sqrt{d+1}}{d} ~~\text{or} ~~  \frac{1}{d} \Rightarrow \beta = 1+ \mathcal{O}(d^{-\frac{1}{2}}).$$
Again, note that there is an equality in the above relationship for the value of $|\braket{\psi^l_i|\psi^m_j}|$ when vectors are from different bases. However, here $\Delta = {0, \frac{1}{d}, \frac{\sqrt{d+1}}{d} }$ is three-valued; it is not APMUB, although the AMUBs are of good quality. The authors have also given a few examples to illustrate the construction.

In \cite{AMUB-FrobeniusRing}, the authors have shown the construction of AMUBs using Gauss sums over Frobenius Rings. In \cite[Lemma 3.2]{AMUB-FrobeniusRing}, for any positive integer, one can construct $p$ many AMUBs over $\mathbb{C}^{\phi(n)}$, where $p$ is the smallest prime divisor of $n$ and $\phi(n)$ is the Euler function. Here
$$|\braket{\psi^l_i|\psi^m_j}| \leq \frac{1}{\sqrt{d}} \left(+  \frac{n - d}{\sqrt{d}(\sqrt{n} +d^\frac{1}{2})}  \right) \Rightarrow \beta = 1+ \mathcal{O}(d^{-\frac{1}{2}}).$$ 
The asymptotic form of $\beta$ is derived assuming $\phi(n) \approx \mathcal{O}(n)$. Here, note that the number of AMUBs is equal to the smallest prime divisor of $n$, which is restrictive and would be small even for $d= p^m$, for $m\geq 2$, as in such cases, complete MUBs are known.

Salient features of all these known constructions are as follows.
\begin{enumerate}
\item All of these methods are based on the some kind of mix of exponential sum and characters of an Abelian groups.

\item All the constructions produce complex AMUBs, i.e., the basis vectors over $\mathbb{C}^d$.

\item The sparsity of all the bases constructed is zero except for the computational basis, whenever they are part of set of AMUB.

\item Though there are combinatorial construction of MUBs in the literature, there is no corresponding combinatorial construction of AMUBs.

\item Most of the constructions produce good quality AMUB only for certain specific forms of dimensions like $d =p-1, q+1, q$, etc. except for \cite{AMUB2}, implying the results and constructions are not applicable for a general set-up.
\end{enumerate}

\section{Theoretical Analysis}
\label{sec:TheoreticalAnalysis}

In this section, we present a generic result dependent on the existence of a suitable RBD, achieved by appropriately categorizing RBD. After that, we explore methods to construct such RBDs. We again emphasize that in the present theoretical analysis, we assume that the points of RBD, i.e., $|X|= d$, can be increased without bound while the parameter $\mu$ remains constant. All our constructions will have this property, justifying the asymptotic analysis of the quality of AMUBs thus constructed. We categorize the analysis into two part, one where all the blocks are of constant size and the other where the blocks are not of constant size in RBD$(X,A)$ with $|X|= d$, a composite number.

\subsection{RBD$(X,A)$ with variable block size.}
In general, RBD$(X, A)$ with $|X|=d$ can have block sizes varying from 1 to $d$. Before presenting the theorem demonstrating how RBD can be used to construct high-quality AMUBs, we provide the following lemma regarding the sparsity of the orthonormal basis constructed using \cite[Construction 1]{ak22}, having different block sizes.

\begin{lemma}
\label{lem:sparsity}
Refer to \cite[Construction 1]{ak22}. If a parallel class $P_l$ of RBD$(X,A)$ has $b$ blocks of sizes $\{k^l_1, k^l_2, \ldots, k^l_b\}$, where $\sum_i k^l_i = |P_l| = |X|=d$,  then the sparsity $(\epsilon)$ of the orthonormal basis constructed using $P_l$ is:
$$ \epsilon = 1 - \frac{{k^l_1}^2+{k^l_2}^2 +\ldots {k^l_b}^2}{d^2} \leq 1- \frac{1}{b}.$$
\end{lemma}

\begin{proof}
To estimate sparsity, refer the construction of an orthonormal basis using RBD$(X,A)$ as in \cite[Theorem 1]{ak22}. Each block within any parallel class, denoted as $P_l$,  consisting of $k_i$ elements, which yields $k_i$ basis vectors. Each of these basis vectors contains $k_i$ many non-zero elements and $(d-k_i)$ zeros. Consequently, a block with $k_i$ elements will contribute $k_i^2$ non-zero elements and $k_i(d-k_i)$ zero elements. Therefore, if a parallel class $P_l$ comprises $b$ blocks of sizes ${ k^l_1, k^l_2,\ldots k^l_t }$, the total number of non-zero components across all the  basis vectors is given by $\sum_i {k^l_i}^2 = {k^l_1}^2 +\ldots +{k^l_t}^2 $. The constraint $\sum_i k^l_i = |P_l| = |X|=d$ represents the total number of elements in the combinatorial design. Under this  constraint, $\sum_i {k^l_i}^2 $ is minimized when $k^l_1 = k^l_2 =\ldots k^l_b = \frac{d}{b}$, resulting in maximum sparsity giving $\sum_i {k^l_i}^2 = \frac{d^2}{b} \Rightarrow \epsilon \leq 1- \frac{1}{b}$. 
\end{proof}

When we know the bounds on the block size of the RBD but do not know the number of blocs, in such situation we can derive bounds on the sparsity using above result, we we state in following corollary 

\begin{corollary}
\label{cor:sparsity}
Refer to \cite[Construction 1]{ak22}. If a parallel class $P_l$ of RBD$(X,A)$, with $|X|= d$, has block sizes bounded below by $k_o$ and and above by $k_m$,  then the sparsity $(\epsilon)$ of the orthonormal basis constructed using $P_l$ 
$$1- \frac{k_m}{d} \leq \epsilon \leq 1- \frac{k_o}{d}.$$
\end{corollary}

\begin{proof}
From Lemma \ref{lem:sparsity}, the $\epsilon = 1 - \frac{{k^l_1}^2+{k^l_2}^2 +\ldots {k^l_b}^2}{d^2}$ where $\sum_i k^l_i = d$. To determine the minimum or maximum value of $\sum_i {k^l_i}^2 $, consider that if $x+y = c$ is a constant, with $x> y$, then $(x+u)^2 +(y-u)^2 > x^2+y^2$. Hence, the maximum value of $\sum_i {k^l_i}^2 $ occurs when the maximum number of $k^l_i$ is as large as possible, while the minimum occurs when maximum number of $k^l_i$ is as small as possible. But since the $k_o \leq k^l_i \leq k_m$, hence  when all the blocks are of size $k_o$, the number of blocks would be $\frac{d}{k_o}$ and when all the blocks would be size $k_o$, the number of blocks would be $\frac{d}{k_m}$. Thus  $\frac{d }{k_o} ~ k_o^2\leq \sum_i {k^l_i}^2 \leq  \frac{d}{k_m}~ k_m^2 \Rightarrow 1- \frac{k_m}{d} \leq \epsilon  \leq 1- \frac{k_o}{d}$.   

%
\end{proof}

Note that $0\leq \epsilon \leq 1- \frac{1}{d}$, where the upper bound corresponds to diagonal unitary matrix, which corresponds to the parallel class having $d$ singleton blocks and the lower bound corresponds to Unitary matrix having no zero entry in it, which corresponds to parallel class having just one block consisting of all the elements of the design. Not all RBD will be useful to the construct set of orthonormal basis as given in \cite[Construction 1]{ak22}, which will also be good quality AMUBs. In this direction, we give the following theorem, where if the block sizes are  $ \mathcal{O}(d^\frac{1}{2}) $, then we can get good quality sparse $\beta$-AMUBs, even when blocks are of different sizes. As noted previously, we call the largest number of elements common between any pair of blocks from different parallel classes the intersection number and denoted as $\mu$.                                                                                                  

\begin{theorem}
\label{thm:d}
Let $(X, A)$ be an RBD with $|X|= d$ and intersection number $\mu$, containing $r$ parallel classes, with block sizes from the set $K =\{q- m, q-m+1, \ldots ,q-m+t\}$, where $q$ is some power of a prime and $m, t \in \mathbb{N}$. Let $q = \sqrt{d} + \eta$, with $\eta \in \mathbb{R}$. If $(m-\eta) \leq \left( \frac{c-\mu}{c}\right) \sqrt{d}$, then we can construct $r$ many $\beta$-AMUBs in dimension $d$, where $\beta =(\frac{\mu}{q-m})\sqrt{d} = \mu + \frac{\mu(m-\eta)}{\sqrt{d}}+\mathcal{O}(d^{-1}) \leq c$. Additionally, the sparsity $(\epsilon)$ is bounded by  Furthermore, if real Hadamard matrices of order equal to every block size exist, then we can construct $r$ many ARMUBs with the same $\beta$ and $\epsilon$.
\end{theorem}

\begin{proof}
We have $|X|=d$ and each parallel class has block size from the set $K =\{q-m, q-m+1,q-m+2,\ldots q-m+t \}$, with $\mu(\geq 1)$ being the maximum number of points being common between any two blocks from different parallel classes. Now, we obtain the maximum value of the dot product between two vectors from different bases would when the vectors are constructed from minimum block sizes. Thus $|\braket{v_1|v_2} | \leq \frac{\mu}{q-m} $ This implies, $\beta = \frac{\mu\sqrt{d}}{q-m}$ and if $\eta = q - \sqrt{d}$, then $\beta \leq c \Rightarrow (m-\eta)\leq \left( \frac{c-\mu}{c}\right) \sqrt{d}$. Since $1\leq \mu \leq c$, we have $(m-\eta) \leq \left( \frac{c-\mu}{c}\right) < 1$. Thus series expansion of $\beta$ in terms of $d$, is given by
$$\beta = \mu \left( 1 + \frac{m- \eta}{\sqrt{d}}+\frac{(m-\eta)^2}{ d} + \frac{(m+\eta)^3}{d^{\frac{3}{2}}}+\ldots \right)$$
showing that $\beta = \mu + \frac{\mu(m-\eta)}{\sqrt{d}}+\mathcal{O}(d^{-1})$.
Now to estimate the sparsity, we use Corollary \ref{cor:sparsity} above. The block size are bounded between $(q-m)$ and $(q-m+t)$, thus, 
$$1- \frac{q-m+t}{d} \leq \epsilon \leq 1- \frac{q-m}{d},$$ 
which implies $\epsilon_o - \frac{t}{d} \leq \epsilon \leq \epsilon_o$, where $\epsilon_o= 1- \frac{q-m}{d}$. And if real Hadamard matrices of order equal to every block size exist, then they can be used as unitary matrices in the  \cite[Construction 1]{ak22} to get real AMUBs with the same parameters, as the choice of Hadamard matrix does not affect the parameters $\beta$ and $\epsilon$ of the constructed AMUBs.
\end{proof}

Note that the parameter $\beta$ is independent of $t$ and depends solely on $\mu$, $m$, and $\eta$. This is because $\beta$, being an upper bound, is determined by the smallest block size of the RBD$(X, A)$ and the intersection number, whereas the sparsity depends on $m$ and $t$, as the block sizes determine it.

\subsection{RBD having constant block size}

In this section, we focus on analyzing the properties of AMUBs constructed using RBDs with a constant block size. The constant block size is essential if we want to utilize Hadamard matrices of the same order (equal to the block size) for all the blocks of the RBD. This is necessary as it ensures that all the basis components finally constructed are either zero or of a constant magnitude $(= 1/\sqrt{k})$, the normalizing factor for each basis vector. If the RBD has a constant block size (denoted as $k$), then the total number of elements, denoted as $|X| = d$, would be a multiple of $k$, thus $d = k \times s$. Now, two situations arise, one when $k \leq s$ and the other when $k > s$, corresponding to the scenarios where block sizes are less than or equal to the number of blocks in a parallel class and vice versa.

The case where $k \leq s$ has been analyzed in detail in \cite{ak23} in constructing APMUBs. For APMUBs, $\mu = 1$ was identified as a necessary condition, and a construction was provided for specific forms of $d = k \times s = (q - e)(q + f), 0 < f \leq e$, where $\mathcal{O}(\sqrt{d})$ many APMUBs could be constructed. The \cite[Theorem 1]{ak23} outlines general features of such AMUBs, concluding that for APMUB construction, one of the crucial requirements was $\mu = 1$, achievable only when $k \leq s$. This is because, when the blocks are of constant sizes, as shown in the previous chapter, $\mu \geq \lceil \frac{k}{s} \rceil$. Hence, for $\mu = 1$, it necessitates that $k \leq s$. 

The constant block size is very useful for constructing ARMUB, as it necessitates the existence of only one real Hadamard matrix of order $k$. Otherwise, a real Hadamard matrix corresponding to all the different block sizes is needed to obtain ARMUB, which is difficult, as real Hadamard matrices are only possible of order $2$ and multiples of $4$. However, if the Hadamard matrix of order $s$ is available and not $k$, we would require RBD to have a block size of order $s$. Thus, in such a situation, we need RBD to have a larger block size than the number of blocks in each parallel class. Toward this, we provide a general result, for $d= k\times s$, on AMUB, constructed using RBD, without assuming $k \leq s$, as was done in \cite[Theorem 4]{ak22}. This can also be viewed as a generalization of the result of \cite[Theorem 4]{ak22} so that a larger class of $d$ can be covered. For this, we state and prove the following theorem.

\begin{theorem}
\label{thm:d=ks}
Suppose there exists an RBD$(X, A)$ with an intersection number $\mu$, having constant block size $k$ and consisting of $r$ parallel classes, where $|X| = d = k \times s$, with $k, s \in \mathbb{N}$ and $2 \delta = (s - k)$ such that $\delta \leq \sqrt{d}$. Then, one can construct $r$ $\beta$-AMUBs in dimension $d$, where $\beta = \mu \sqrt{\frac{s}{k}} = \mu\left(1 + \frac{\delta}{\sqrt{d}} + \mathcal{O}(d^{-1}) \right) \leq (1+\sqrt{2})\mu$, and sparsity $\epsilon = 1 - \frac{1}{s}$. Furthermore, if there exists a real Hadamard matrix of order $k$, we can construct $r$ APRMUBs with the same $\beta$ and $\epsilon$ values and $\Delta = {0, \frac{1}{k}, \frac{2}{k}, \ldots, \frac{\mu}{k}}$. 
\end{theorem}

\begin{proof}
We have $|X|=d= k\times s$, where each parallel class has a block size of $k$ and $\mu$ represents the maximum number of points common between blocks from different parallel classes. Consequently, $\beta = \frac{\mu \sqrt{d}}{k} = \mu \sqrt{\frac{s}{k}}$. If $2\delta = s- k$, then $d = (s- 2 \delta ) s $, and solving for $s$ and $k$, we obtain $s= \sqrt{d+\delta^2} +\delta$ and $k= \sqrt{d+\delta^2} -\delta$. Therefore, if $\beta = \mu \sqrt{\frac{s}{k}} \leq c$, it follows that $\delta \leq \left(\frac{c^2-\mu^2}{2\mu c}\right) d^\frac{1}{2}$. Assuming $\delta $ to be small and bounded, we have $\beta = \mu(1+ \frac{\delta}{\sqrt{d}} + \frac{\delta ^2}{2d} +\ldots)$.

The sparsity, as defined in Lemma \ref{lem:sparsity}, is given by $\epsilon = 1 - \frac{k}{d} = 1 - \frac{1}{s}$. If a real Hadamard matrix of order $k$ exists, we can use them in the construction to obtain $r$ many real approximate MUBs in $\mathbb{R}^d$, with the same $\beta$ and sparsity $(\epsilon)$. In the case of real approximate MUBs over $\mathbb{R}^d$, the $\Delta$ values, representing different possible absolute values of the inner product, would be $\Delta = \{0, \frac{1}{k}, \frac{2}{k}, \ldots \frac{\mu}{k} \}$. This implies that in the case of ARMUBs constructed using RBD$(X, A)$ with a constant block size $k$, $|\Delta|= \mu +1$, and since $\mu$ is usually a small positive integer, $\Delta$ constitutes a small set.
\end{proof}

The form of $d$ as $(q-e)(q+f)$ in \cite{ak23} was chosen because when $q$ is some power of a prime, using $(q^2,q,1)$-RBIBD, one could construct having $\mathcal{O}(\sqrt{d})$ many parallel classes. In \cite{ak23}, the focus was to construct RBD$(X,A)$ such that $\mu = 1$, hence the block size could only be $(q-e)$ and not $(q+f)$. Not all composite $d$ could be written in this form as it required the existence of prime power $q$, such that $\frac{k+s}{2} \leq q \leq s$. Nevertheless, expressing $d=(q-e)(q+f)$ proved beneficial in increasing the number of AMUB to $\mathcal{O}(\sqrt{d})$. In the similar direction of constructing $\mathcal{O}(\sqrt{d})$ many parallel classes, for a larger class of composite $d$, we now express composite $d$ as $(q-e)(q\pm f), 0< f \leq e$, where $q$ is some suitable prime power depending on factors of $d$ and focus on constructing RBD$(X, A)$ with $|X|=d$ and having $q$ many parallel classes, with intersection number $\mu$. The condition $0< f \leq e$ will ensure $(q-e) \leq (q\pm f)$. Also note that when $d= q^2$, i.e., corresponding to $e=f=0$, the RBD method gives $q+1$ many Orthonormal bases, which are MUBs \cite[Corollary 3 and Corollary 4]{ak22}. Hence, expressing a composite $d$ which is not some power of a prime number, as $(q-e)(q\pm f), 0< f \leq e$ for small values of $e$ and $f$, will also help in understanding how the small perturbation in $q$ affects the nature of the constructed orthonormal basis and its deviation from MUBs.   

Using the theorem above along with \cite[Theorem ?]{ak23}, and expressing $d = k \times s = (q - e)(q \pm f)$, where $0 < f \leq e$, we summarize the results, providing the values of its $\beta$ and sparsity $\epsilon$ for the AMUB constructed under various situations for RBD$(X, A)$, where $|X| = d$, and all blocks of the design are of constant size, depicted as follows.

\begin{center}
\begin{tabular}{ |c|c|c|c|c|c| } 
\hline
$d= k\times s$ & Block size &$\mu_{min}$ & $\beta$  & $\epsilon$ & $\beta/\mu$ \\
\hline
\hline 
$(q-e)(q+f)$ & $(q-e)$ & 1& $\mu \sqrt{\frac{q+f}{q-e}}$  & $1 - \frac{1}{q+f}$ & 1+ $\frac{e+f}{2\sqrt{d}} + \ldots$\\ 
\hline
$(q-e)(q+f)$ & $(q+f)$ &2& $\mu \sqrt{\frac{q-e}{q+f} }$  & $1 - \frac{1}{q-e}$ & $1-  \frac{e+f}{2\sqrt{d}} + \ldots$\\ 
\hline
$(q-e)(q-f)$ & $q-e$ & 1&$ \mu \sqrt{\frac{q-f}{q-e} }$ & $1 - \frac{1}{q- f}$ & $1+ \frac{e-f}{2\sqrt{d}} + \ldots$\\ 
\hline
$(q-e)(q-f)$ & $q-f$ &2 &$\mu \sqrt{\frac{q-e}{q-f}}$  & $1 - \frac{1}{q-e}$ & $1- \frac{e-f}{2\sqrt{d}} + \ldots$\\ 
\hline
\end{tabular}
\end{center}

We have ignored the higher-order terms in the expansion of $\beta$ as a function of $1/\sqrt{d}$. Furthermore, it is worth noting that $(e+f)$ and $(e-f)$, which appear above, are equivalent to $2 \delta = s-k$ for their respective cases. Although in the cases where $\mu_{\text{min}} = 2$, APMUBs cannot be obtained, it still qualifies as a $\beta$-AMUB, with $\beta$ bounded above by $\mu$. Therefore, for constructing high-quality $\beta$-AMUBs, it is desirable to use RBD$(X, A)$, ensuring that $\mu$ is minimized as much as possible. In this connection let us following lemma on Absolute Lower Bound on $\beta$

\begin{lemma}
Let RBD(X,A) with $|X|= d= k\times s$ having constant block size $k$, is used to construct $\beta$-AMUB where each parallel class has $s$ many blocks, then $\beta \geq \sqrt{\frac{k}{s}}$.
\end{lemma}

\begin{proof}
From \cite[Lemma 2]{ak23} for $d= k\times s$, where $k$ is the block size and $s$ is the number of blocks in a parallel class, then $\mu \geq \lceil \frac{k}{s} \rceil \geq \frac{k}{s}$. Thus, we have $\beta = \mu \sqrt{\frac{s}{k}} \geq \sqrt{\frac{k}{s}}$
\end{proof}
Note that $\beta \geq 1$ in all cases, thus this provides a better lower bound for  $\beta$ than the trivial lower bound of $1$,  in situations where the block size exceeds the number of blocks i.e., $k> s$, in the RBD$(X, A)$, having a constant block size.

We now state some facts about the gaps between two consecutive primes. In ~\cite{baker2001difference} it was shown that there is prime in interval$[x-x^\theta, x]$ for $x$ greater than sufficiently large  integer say $n_o$ where  $\theta = 0.525$. Hence, for any two consecutive prime we have, $p_{n+1} - p_n = \mathcal{O}(p_n^{0.525})$ for all prime, larger than  $n_o$.  If we define $g_n = p_{n+1} - p_n $ then the ratio $\frac{g_n}{Log(p_n)}$ is known as merit of the gap $g_n$. There are only few large values of this ratio. There is another important figure of merit for gap between the consecutive prime, is called Cram\'er - Shanks - Granville ratio based on Cram\'er conjecture \cite{cramer1936order}. It is defined as the ratio $\frac{g_n}{Log(p_n)^2}$. Shanks conjectured that this ratio will always be less than $1$, where as Granville conjectured, that the ratio will exceed $1$ or come arbitrarily close to $2/e^\gamma = 1.1229$ \cite{leclair2015asymptotic,sinha2010new}. On the other hand Firoozbakht's conjecture implies that the ratio is below $1-\frac{1}{\log(p)}$ for all primes $p \geq11$ \cite{kourbatov2015upper}. The greatest known value of this ratio is about $0.92$, after discarding the anomalously high values of the ratio for the small primes less than or equal to $7$. Thus assuming Cram\'er conjecture, there is a prime number in interval $[x, x+\varrho \log^2 x]$ for all $x \geq 7$ where currently all the  known value of $\varrho < 1 $.\\

We now state and prove the following lemma related to expressing any composite $d= k\times s = (q-e)(q\pm f)$ and examine the availability of such $q$ to express $d$ in the above form and examine the asymptotic dependence of $\beta$ on $q, e, f$, and get an estimate of these values in terms of $d$ and its factors $k$ and $s$. These will help in analyzing constructions in the next section. We state and prove the following lemmas.

\begin{lemma}
\label{lem:e+f}
If $d = k \times s$ such that $\delta = \frac{s - k}{2} \geq s^\theta$, then there exists a prime power $q$ such that $d = (q - e)(q + f)$ and $0 \leq f \leq \delta \leq e \leq 2\delta$. Here, $\theta = 0.525$ for sufficiently large $s$.
\end{lemma} 
 
\begin{proof}
We are looking for $q$ a power of prime, between $k$ and $s$ such that $d= (q-e)(q+f)$ with $0\leq f\leq e \Rightarrow s-q \leq q-k \Rightarrow \frac{s+k}{2} \leq q$. Hence $ \frac{s+k}{2} \leq q \leq s$. The unconditional result on the gaps in the prime \cite{baker2001difference} implies the existence of a prime number in the interval $[s-s^\theta, s]$ for sufficiently large $s$. The current known value of $\theta = 0.525$. Applying this, we get $s-\frac{s+k}{2} =  \delta \geq s^\theta $ as a sufficient condition for the existence of such prime power $q$.  Since $\frac{k+s}{2} \leq q \leq s$, hence $\delta = \frac{k+s}{2} -k \leq e = q-k\leq s-k = 2\delta$ and $0\leq f= s-q \leq s-\frac{k+s}{2} = \delta $. Thus $0 \leq f \leq \delta \leq e \leq 2\delta$ such that $\frac{e+f}{2} =\delta$.
\end{proof}

Now assuming Cram\'er conjecture \cite{cramer1936order} on the gap in prime number in terms of Cram\'er - Granville - Shanks ratio, which is less than 1, we have $\delta \geq \varrho \log^2 s$ as a sufficient condition for the existence of such a prime power $q$. Therefore, if it is impossible to find such prime between $k$ and $s$ only in situation  when there is no 'sufficient' gap between them. In such a situation, we find $q$ closest to $s$ but greater than $s$ and express $d = (q - e)(q - f)$. In this direction, we have the following lemma.

\begin{lemma}
\label{lem:e-f}
Let $d= k\times s$, and there is no prime power between $k$ and $s$, then $d= \mathcal{O}(s^2)$ and we can express $d= (q-e)(q-f)$, with $q$ being some prime power greater than $s$, such that $f= \mathcal{O}(d^{\frac{\theta}{2}})$ and $e = 2\delta + \mathcal{O}(d^{\frac{\theta}{2}}) $ where $\delta = \frac{s-k}{2}$ and $\theta = 0.525$.
\end{lemma}

\begin{proof}
The result on prime power gaps states that for sufficiently large $x$, there is a prime number between $[x-x^\theta, x]$. Given that there is no prime power between $k$ and $s$, this implies $(s-k) = 2 \delta < s^\theta$. In such a situation, choose the smallest prime power $q$ such that $q \geq s$ and express $d = k \times s = (q-e)(q-f)$ where $f= \mathcal{O}(s^\theta)$ and $e = 2\delta + \mathcal{O}(s^\theta)$. Further note that in this situation, since $2 \delta \leq s^\theta \Rightarrow k \geq s- s^\theta$. Thus $(s-s^\theta)s \leq d \leq s^2 \Rightarrow s^2 - s^{1+\theta} \leq d \leq s^2$. Hence $d= \mathcal{O}(s^2)$. Also note that $d= (s- 2\delta)s \Rightarrow s= d^\frac{1}{2}+ \delta + \frac{\delta^2}{4\sqrt{d}}+..$ and since $\delta \leq s^\theta$, hence $s=\mathcal{O}(\sqrt{d}) $ thereby implying $f= \mathcal{O}(d^{\frac{\theta}{2}})$ and $e = 2\delta + \mathcal{O}(d^{\frac{\theta}{2}})$.
\end{proof}

Note that in this case $\frac{e+f}{2} =\delta + f = \delta + \mathcal{O}(d^{\frac{\theta}{2}}) $. And again if we assume Cram\'er conjecture \cite{cramer1936order} on the gap in prime number, then in terms of Cram\'er - Granville - Shanks ratio we have $f= \mathcal{O}(\log^2 s) $ and $e = 2\delta + \mathcal{O}(\log^2 s)$ and $\frac{e+f}{2} =\delta + \mathcal{O}(\log^2 s)$.

\section{Construction of AMUB through RBD}
\label{sec:4}
We now present the construction of several sets of $\beta$-AMUBs, which could be useful in information processing in classical and quantum domains. We demonstrate that for a given $d$, there can be distinct $\beta$-AMUBs with different characteristics. We will compare the parameters of our construction with known constructions of AMUBs, illustrating and highlighting salient features of the present construction and how it surpasses known AMUBs in certain aspects.

To construct AMUBs through RBD, we proceed in two steps. First, we construct suitable RBDs. Then, we using the steps of  \cite[Construction 1]{ak22} to construct a set of unitary matrices corresponding to each parallel class of RBD, with parameters following the results provided in Section \ref{sec:TheoreticalAnalysis}. For ease of understanding, we will initially demonstrate each construction with a simple example, followed by a general algorithm for construction and proof of the correctness of the construction in the form of a lemma. This will be followed by  Theorem/Lemma concerning APMUB for the given form of $d$ for which the construction has been demonstrated.

We present this in two parts: one devoted to constructing AMUBs using RBDs with non-constant block size and another devoted to constructing AMUBs using RBDs with constant block size. 

\subsection{Construction of $\beta$-AMUB through RBD having non-constant block size}

In this section, we demonstrate that for any composite $d = k \times s$ such that $|k-s| < \sqrt{d}$, we can construct an RBD$(X,A)$ with $\mu = 1$ containing more than $\sqrt{d}$ parallel classes with non-constant block sizes and fetching $\mathcal{O}(\sqrt{d})$ sparse $\beta$-AMUBs with $\beta = 1 + \mathcal{O}(d^{-\frac{1}{2}})$, indicating that for large $d$, it approaches MUBs. To achieve this, we express $d$ as either $d = (q-e)(q+f)$ or $(q-e)(q-f)$, where $q$ is some prime power.


Let us first consider the case for $d = (q - e)(q + f)$ with $0 < f \leq e$. In \cite[Theorem 3]{ak23} we have shown that when $d$ is of the form $(q - e)(q + f)$, then we can construct $\lfloor \frac{q - e}{f} \rfloor + 1$ many APMUBs, with $\beta = 1 + \mathcal{O}(d^{-\frac{1}{2}})$. Toward the proof of this theorem, the \cite[Construction 4]{ak23} and corresponding \cite[Lemma 7]{ak23} shows the existence of RBD$(\widetilde{X},\widetilde{A})$, where $|\tilde{X}|= d= (q-e)(q+f)$ with $f\leq e$ having $r=q+1$ parallel classes, and $\mu =1$. The block sizes are from set $\{(q-e),(q-e+1),\ldots (q-e+f) \}$. Now in Theorem \ref{thm:d=ks},  one can use RBD$(\widetilde{X},\widetilde{A})$ to construct $q + 1$ many $\beta$-AMUBs. But since the block sizes are not constant, hence $\Delta$ will consist of more than two elements, and thus it will not satisfy the criteria of APMUB, even though $\beta = 1+ \mathcal{O}(d^{-\lambda})$ with $\lambda= \frac{1}{2}$ and $\beta \leq 2$ if $0\leq (e+f) \leq \frac {3}{2} \sqrt{d}$. To characterize the $\Delta$ of resulting AMUB from RBD$(\widetilde{X},\widetilde{A})$,  we define $\Delta_1$ as follows.

\begin{definition}
\label{def:delta1}
For $d=(q-e)(q+f)$, where $e\geq f$ and let $\theta_1 = \frac{1}{\sqrt{q-e}}, \theta_2 = \frac{1}{\sqrt{q-e+1}},\ldots ,\theta_{f+1} = \frac{1}{\sqrt{q-e+f}}, \theta_{f+2} = \frac{1}{\sqrt{q}} $ then define $\Delta_1 = \{\theta_i \theta_j \} \cup \{0\}$ where  $i,j = 1,2,\ldots,(f+1),(f+2)$. 
\end{definition}

Note that  $|\Delta_1| = {f+2\choose 2} + (f+2) +1 = \frac{(f+3)(f+2)}{2}+1 = \mathcal{O}(f^2)$. Thus the number of elements in $\Delta_1$ is only dependent on the value of $f$ and is proportional to the square of it. We now state and prove the result on $\beta$-AMUB, using such RBD$(X,A)$ as follows

\begin{corollary}
\label{cor:(q-e)(q+f)}
If $d=(q-e)(q+f)$, for some prime-power $q$, and $e,f\in\mathbb{N}$ satisfying $0<f\leq e $ and $0<(e+f) \leq \frac {3}{2} \sqrt{d}$, then there exist at least $r = q+1$ many $\beta$-AMUBs, with $\Delta \subseteq \Delta_1$, $\beta = \sqrt{\frac{q+f}{q-e} } =1 +\frac{e+f}{2\sqrt{d}} + \mathcal{O}(d^{-1}) \leq 2$, and $1- \frac{q-e+f}{d} \leq \epsilon \leq 1- \frac{1}{q}$ where $\epsilon $ denotes the sparsity.
\end{corollary}

\begin{proof}
Consider the RBD$(\widetilde{X},\widetilde{A})$ as constructed in \cite[Construction 4]{ak23}, and its property given in \cite[Lemma 7]{ak23}. The $|\widetilde{X}|= (q-e)(q+f)$ with $f \leq e$. The block sizes of RBD$(\widetilde{X},\widetilde{A})$ are from the set ${(q-e),(q-e+1),\ldots, (q-e+f), q }$. Now using Hadamard matrices of the order of the block sizes, we get the set of $q+1$ orthonormal basis. Thus, the normalizing factors of Hadamard matrices of different orders would be from set $\left\{\frac{1}{\sqrt{q-e}},\frac{1}{\sqrt{q-e+1}},\ldots, \frac{1}{\sqrt{q-e+f}}, \frac{1}{\sqrt{q}}\right\}$. Let's denote this set of factors by $S_\theta= \left\{\theta_1,\theta_2, \ldots, \theta_{f+2} \right\}$ where $\theta_1 = \frac{1}{\sqrt{q-e}}$, $\theta_2 = \frac{1}{\sqrt{q-e+1}}$, $\ldots$, $\theta_{f+1} = \frac{1}{\sqrt{q-e+f}}$, and $\theta_{f+2} = \frac{1}{\sqrt{q}}$. The set will have  $f+2$ elements. Since $\mu = 1$ for RBD$(\widetilde{X},\widetilde{ A})$, we have $|\braket{\psi^l_i|\psi^m_j}| = \theta_i \theta_j$ or 0, where $\ket{\psi^l_i}$ and $\ket{\psi^m_i}$ are the vectors from two different orthonormal bases constructed using the parallel class of RBD$(\widetilde{X}, \widetilde{A})$ and $\theta_i, \theta_j \in S_\theta$. Thus, $\Delta \subseteq \Delta_1$.

Now, $\max \{ |\braket{\psi^l_i|\psi^m_j}| \} = \max \{ \theta_i \theta_j \} = \frac{1}{q-e} $ corresponding to the maximum value of $\theta_i = \theta_j = \frac{1}{\sqrt{q-e}}$. Thus, $\beta = \frac{d}{q-e}= \sqrt{\frac{q+f}{q-e} } =1 +\frac{e-f}{2\sqrt{d}} + \mathcal{O}(d^{-1}) $. And for $\beta \leq 2$ as stated in \cite[Theorem 3]{ak23}  we should have  $0 < (e+f) \leq \frac {3}{2} d^{\frac{1}{2}}$.

To estimate the value of sparsity, refer to the Lemma \ref{lem:sparsity}, since there are $q$ blocks in each parallel class, thus Sparsity$(\epsilon)$ is bounded above by $1- \frac{1}{q}$ and using Corollary \ref{cor:sparsity}, the lower bound of $\epsilon$ is $1- \frac{q-e+f}{d}$. Hence  $ 1- \frac{q-e+f}{d} \leq \epsilon \leq 1- \frac{1}{q}$.
\end{proof}

Note the that here we have improved the number of $\beta$-AMUBs for dimensions of the for  $d= (q-e)(q+f)$ from $\lfloor \frac{q-e}{f} \rfloor +1$ to $(q+1)$ many $\beta$-AMUBs with the same $\beta = 1 + \mathcal{O}(d^{-\frac{1}{2}})$. However, the order of set $\Delta$ has now increased from two valued viz. $\left\{0, \frac{\beta}{\sqrt{d}}\right\}$ to $\Delta_s$ as in Definition \ref{def:delta1}, thus it is not APMUB. Also as stated in  \cite[Theorem 3]{ak23}, if there exists a real Hadamard matrix of order $(q-e)$, then there exists at least $r= \lfloor \frac{q-e}{f} \rfloor +1 $ many Almost Perfect Real MUBs (APRMUBs). However, the same cannot be applicable here  for Corollary \ref{cor:(q-e)(q+f)} as the block sizes vary from $(q-e)$ to $(q-e+f)$ and $q$, and hence for ARMUBs, we require real Hadamard matrices of all these orders. Since a Hadamard matrix can only exist when the order is divisible by 4, it is not possible to construct ARMUBs using RBD$(\widetilde{X}, \widetilde{A})$. Also the sparsity in the above case  is bounded above by $1 - \frac{1}{q}$, whereas the sparsity in the case of APMUB \cite[Theorem 3]{ak23}  is $1 - \frac{1}{q+f}$. Hence, for $f > 0$, the APMUBs are sparser than the AMUBs constructed here for same $d= (q-e)(q+f)$. 

For example, when the RBD$(X,A)$ with $d= 4\times 8 = (7-3)(7+1)$, is used for constructing $\beta$-AMUB, we obtain $\Delta = \{ \frac{\beta_i}{\sqrt{d}} \}$ where set $\beta_i = \{0,~0.79, ~0.96, ~1.07, ~1.13, ~1.24, ~1.41 \}$
Thus, note that $\beta = \sqrt{\frac{q+f}{q-e}} = \sqrt{2} =1.41$ in this case, and the sparsity $\epsilon \leq 1- \frac{1}{7}= 0.86$. On the other hand, the APMUB for the same $d= 4\times 8$, we obtain $\Delta = \{ 0, \frac{\beta}{\sqrt{d}} \} = \{ 0, 1.41 \} $ and sparsity $\epsilon = 1- \frac{1}{7+1}=0.88 $. Thus, the $\Delta$ reduces to two elements, and the sparsity increases slightly. Note that, $\beta = \sqrt{\frac{q+f}{q-e}} = \sqrt{2} =1.41 $ remain same for both the cases. But the advantage here is that we have 8 many $\beta$-AMUB where as there were only $5$ APMUB. 


As we have noted that when there is not sufficient gap between  $k$ and $s$, there need not be any prime power $q$ between $k$ and $s$ using which we can express $d= (q-e)(q+f)$ with $f\leq e$. In such situation we find $q$ greater than $s$ and express $d= (q-e)(q-f)$ for suitable $e$ and $f$. Refer Lemma \ref{lem:e+f} and Lemma \ref{lem:e-f} above related to it. 

Let us now consider $d = (q-e)(q-f)$, where $q$ is a prime power with $0 < f \leq e$. First, we demonstrate that in such cases, we can construct an RBD$(X, A)$ with $|X| = d$ having $q$ many parallel classes, where  block sizes  in the parallel classes are from  set $\left\{ q-(e+f), q-(e+f)+1,\ldots, q-e\right\}$  Consequently, such an RBD$(X, A)$ can be utilized to construct $q$ orthonormal bases following  Theorem \ref{thm:d}  thus providing $\mathcal{O}(q)$ AMUBs in such scenarios.

For constructing such an RBD, we consider a $(q^2, q, 1)$-Affine Resolvable BIBD as the input. We denote this RBD$(\bar{X}, \bar{A})$, where $|\bar{X}|=q^2$ and all blocks of $A$ have the same size $q$, with the number of parallel classes in $A$ being $q+1$. Utilizing this, we construct RBD$(X, A)$, where $|X| = (q-e)(q-f)$ with the same number of parallel classes $q$, but the blocks do not have the same size.

Let us illustrate this construction with an example. Consider $|X| = (7-2)(7-1) = 5 \times 6 = 30$, where $q=7$, $e=2$, and $f=1$. We employ an Affine Resolvable $(7^2,7,1)$-BIBD, referred to as RBD$(\bar{X}, \bar{A})$, comprising eight parallel classes. Each parallel class comprises seven blocks of constant size 7. We represent each parallel class as a $7 \times 7$ matrix, with each row representing one block of the parallel class. Therefore, there would be $8$ such matrices as shown below to depict the combinatorial design.

\small{
\begin{equation*}
\bar{P}_1 = \begin{bmatrix}
\bar{b}^1_7= \{ 1 & 2 & 3 & 4 & 5 & 6 & 7 \} \\
\bar{b}^1_6= \{ 8 & 9 & 10 & 11 & 12 & 13 & 14 \}\\
\bar{b}^1_5= \{ 15 & 16 & 17 & 18 & 19 & 20 & 21 \}\\
\bar{b}^1_4= \{ 22 & 23 & 24 & 25 & 26 & 27 & 28 \}\\
\bar{b}^1_3= \{ 29 & 30 & {\color{red}31} & {\color{red}32} & {\color{red}33} & {\color{red}34} & {\color{red}35}  \}\\
\bar{b}^1_2= \{ {\color{red}36} & {\color{red}37} & {\color{red}38} & {\color{red}39} & {\color{red}40} & {\color{red}41} & {\color{red}42} \}\\
\bar{b}^1_1= \{ {\color{red}43} & {\color{red}44} & {\color{red}45} & {\color{red}46} & {\color{red}47} & {\color{red}48} & {\color{red}49} \}\\
\end{bmatrix},
\bar{P}_2 = \begin{bmatrix}
\bar{b}^2_7= \{ 1 & 9 & 17 & 25 & {\color{red}33} & {\color{red}41} & {\color{red}49} \}\\
 \bar{b}^2_6= \{ 2 & 10 & 18 & 26 & {\color{red}34} & {\color{red}42} & {\color{red}43}\} \\
\bar{b}^2_5= \{  3 & 11 & 19 & 27 & {\color{red}35} & {\color{red}36} & {\color{red}44} \}\\
 \bar{b}^2_4= \{ 4 & 12 & 20 & 28 & 29 & {\color{red}37} & {\color{red}45} \}\\
 \bar{b}^2_3= \{ 5 & 13 & 21 & 22 & 30 & {\color{red}38} & {\color{red}46} \}\\
 \bar{b}^2_2= \{ 6 & 14 & 15 & 23 & {\color{red}31} & {\color{red}39} & {\color{red}47} \}\\
 \bar{b}^2_1= \{ 7 & 8 & 16 & 24 & {\color{red}32} & {\color{red}40} & {\color{red}48}  \}\\
 \end{bmatrix},
\end{equation*}

\begin{equation*}
\bar{P}_3 = \begin{bmatrix}
\bar{b}^3_7= \{  1 & 10 & 19 & 28 & 30 & {\color{red}39} & {\color{red}48} \}\\
\bar{b}^3_6= \{  2 & 11 & 20 & 22 & {\color{red}31} & {\color{red}40} & {\color{red}49} \}\\
\bar{b}^3_5= \{  3 & 12 & 21 & 23 & {\color{red}32} & {\color{red}41} & {\color{red}43} \}\\
\bar{b}^3_4= \{  4 & 13 & 15 & 24 & {\color{red}33} & {\color{red}42} & {\color{red}44} \} \\
\bar{b}^3_3= \{  5 & 14 & 16 & 25 & {\color{red}34} & {\color{red}36} & {\color{red}45}\}\\
 \bar{b}^3_2= \{ 6 & 8 & 17 & 26 & {\color{red}35} & {\color{red}37} & {\color{red}46}  \}\\
\bar{b}^3_1= \{ 7 & 9 & 18 & 27 & 29 & {\color{red}38} & {\color{red}47} \}\\
 \end{bmatrix},
\bar{P}_4 = \begin{bmatrix}
\bar{b}^4_7= \{  1 & 11 & 21 & 24 & {\color{red}34} & {\color{red}37} & {\color{red}47} \}\\
\bar{b}^4_6= \{  2 & 12 & 15 & 25 & {\color{red}35} & {\color{red}38} & {\color{red} 48} \}\\
\bar{b}^4_5= \{  3 & 13 & 16 & 26 & 29 & {\color{red}39} & {\color{red}49} \}\\
\bar{b}^4_4= \{  4 & 14 & 17 & 27 & 30 & {\color{red}40} & {\color{red}43} \}\\
\bar{b}^4_3= \{  5 & 8 & 18 & 28 & {\color{red}31} & {\color{red}41} & {\color{red}44} \}\\
\bar{b}^4_2= \{  6 & 9 & 19 & 22 & {\color{red}32} & {\color{red}42} & {\color{red}45} \}\\
\bar{b}^4_1= \{  7 & 10 & 20 & 23 & {\color{red}33} & {\color{red}36} & {\color{red}46} \}\\
 \end{bmatrix},
\end{equation*}

\begin{equation*}
\bar{P}_5 = \begin{bmatrix}
\bar{b}^5_7= \{      1 & 12 & 16 & 27 & {\color{red}31} & {\color{red}42} & {\color{red}46} \}\\
\bar{b}^5_6= \{ 2 & 13 & 17 & 28 & {\color{red}32} & {\color{red}36} & {\color{red}47} \}\\
\bar{b}^5_5= \{  3 & 14 & 18 & 22 & {\color{red}33} & {\color{red}37} & {\color{red}48} \}\\
\bar{b}^5_4= \{  4 & 8 & 19 & 23 & {\color{red}34} & {\color{red}38} & {\color{red}49} \}\\
\bar{b}^5_3= \{  5 & 9 & 20 & 24 & {\color{red}35} & {\color{red}39} & {\color{red}43} \}\\
\bar{b}^5_2= \{  6 & 10 & 21 & 25 & 29 & {\color{red}40} & {\color{red}44} \}\\
\bar{b}^5_1= \{  7 & 11 & 15 & 26 & 30 & {\color{red}41} &{\color{red}45} \}\\
 \end{bmatrix},
\bar{P}_6 = \begin{bmatrix}
\bar{b}^6_7= \{      1 & 13 & 18 & 23 & {\color{red}35} & {\color{red}40} & {\color{red}45} \}\\
\bar{b}^6_6= \{ 2 & 14 & 19 & 24 & 29 & {\color{red}41} &{\color{red}46} \}\\
\bar{b}^6_5= \{ 3 & 8 & 20 & 25 & 30 & {\color{red}42} & {\color{red}47} \}\\
\bar{b}^6_4= \{ 4 & 9 & 21 & 26 & {\color{red}31} & {\color{red}36} & {\color{red}48} \}\\
\bar{b}^6_3= \{ 5 & 10 & 15 & 27 & {\color{red}32} &{\color{red}37} & {\color{red}49} \}\\
\bar{b}^6_2= \{ 6 & 11 & 16 & 28 & {\color{red}33} & {\color{red}38} & {\color{red}43} \}\\
\bar{b}^6_1= \{ 7 & 12 & 17 & 22 & {\color{red}34} & {\color{red}39} &{\color{red}44} \}\\
 \end{bmatrix},
\end{equation*}

\begin{equation*}
\bar{P}_7 = \begin{bmatrix}
\bar{b}^7_7= \{1 & 14 & 20 & 26 & {\color{red}32} & {\color{red}38} & {\color{red}44} \}\\
\bar{b}^7_6= \{ 2 & 8 & 21 & 27 & {\color{red}33} & {\color{red}39} & {\color{red}45} \}\\
\bar{b}^7_5= \{ 3 & 9 & 15 & 28 & {\color{red}34} & {\color{red}40} & {\color{red}46} \}\\
\bar{b}^7_4= \{ 4 & 10 & 16 & 22 & {\color{red}35} & {\color{red}41} & {\color{red}47} \}\\
\bar{b}^7_3= \{ 5 & 11 & 17 & 23 & 29 & {\color{red}42} & {\color{red}48} \}\\
\bar{b}^7_2= \{ 6 & 12 & 18 & 24 & 30 & {\color{red}36} & {\color{red}49} \} \\
\bar{b}^7_1= \{ 7 & 13 & 19 & 25 & {\color{red}31} &{\color{red}37} & {\color{red}43} \}\\
 \end{bmatrix},
 \bar{P}_8 = \begin{bmatrix}
\bar{b}^8_7= \{1 & 8 & 15 & 22 & 29 & {\color{red}36} & {\color{red}43} \}\\
\bar{b}^8_6= \{ 2 & 9 & 16 & 23 & 30 & {\color{red}37} & {\color{red}44} \}\\
\bar{b}^8_5= \{ 3 & 10 & 17 & 24 & {\color{red}31} &{\color{red}38} & {\color{red}45} \}\\
\bar{b}^8_4= \{ 4 & 11 & 18 & 25 & {\color{red}32} & {\color{red}39} & {\color{red}46} \}\\
\bar{b}^8_3= \{ 5 & 12 & 19 & 26 &{\color{red} 33} & {\color{red}40} & {\color{red}47} \}\\
\bar{b}^8_2= \{ 6 & 13 & 20 & 27 & {\color{red}34} & {\color{red}41} & {\color{red}48} \}\\
\bar{b}^8_1= \{ 7 & 14 & 21 & 28 & {\color{red}35} & {\color{red}42} & {\color{red}49} \}\\
 \end{bmatrix},
\end{equation*}
}
\normalsize

In order construct RBD$(X,A)$, where $|X| = (q-e)(q-f)= (7-2)(7-1)= 30$ such that $\mu = 1$, do the following. 

\begin{enumerate}
\item Choose any $e (= 2)$ many blocks from $\bar{P}_1$. Let these blocks be $\bar{b}^1_1$ and $\bar{b}^1_2$. Let $S_1 = \bar{b}^1_1 \cup \bar{b}^1_2 =  \{ 36, 37, 38, 39, 40, 41, 42, 43, 44, 45, 46, 47, 48, 49 \}$.

\item  Choose another ($f = 1$) block from $\bar{P}_1$. Let it be $\bar{b}^1_3$. Now choose any $(q-e)= 5$ elements from each of the $f$ blocks. Let these be $S_2 = \{31,32,33, 34, 35\}$. Set $S = S_1 \cup S_2 = \{31,32,33, 34, 35, 36, 37, 38, 39, 40, 41, 42$, $43, 44, 45, 46, 47, 48, 49 \} $ (indicated in red).

\item  Remove the elements of the set $S$ from the RBD $(\bar{X},\bar{A})$. Call the resulting combinatorial design a new RBD $(X,A)$, where $|X|=30$ and $A = \{P_1, P_2, P_3, P_4, P_5, P_6, P_7, P_8 \}$, presented as below.

\small
\begin{equation*}
P_1 = \begin{bmatrix}
b^1_7= \{ 1 & 2 & 3 & 4 & 5 & 6 & 7 \} \\
b^1_6= \{ 8 & 9 & 10 & 11 & 12 & 13 & 14 \}\\
b^1_5= \{ 15 & 16 & 17 & 18 & 19 & 20 & 21 \}\\
b^1_4= \{ 22 & 23 & 24 & 25 & 26 & 27 & 28 \}\\
b^1_3= \{ 29 & 30    \}\\
\end{bmatrix},
P_2 = \begin{bmatrix}
b^2_7= \{ 1 & 9 & 17 & 25   \}\\
b^2_6= \{ 2 & 10 & 18 & 26 \} \\
b^2_5= \{  3 & 11 & 19 & 27   \}\\
b^2_4= \{ 4 & 12 & 20 & 28 & 29  \}\\
b^2_3= \{ 5 & 13 & 21 & 22 & 30  \}\\
b^2_2= \{ 6 & 14 & 15 & 23  \}\\
b^2_1= \{ 7 & 8 & 16 & 24   \}\\
\end{bmatrix},
\end{equation*}
 
\begin{equation*}
P_3 = \begin{bmatrix}
b^3_7= \{  1 & 10 & 19 & 28 & 30  \}\\
b^3_6= \{  2 & 11 & 20 & 22  \}\\
b^3_5= \{  3 & 12 & 21 & 23   \}\\
b^3_4= \{  4 & 13 & 15 & 24  \} \\
b^3_3= \{  5 & 14 & 16 & 25  \}\\
 b^3_2= \{ 6 & 8 & 17 & 26    \}\\
b^3_1= \{ 7 & 9 & 18 & 27 & 29 \}\\
 \end{bmatrix},
P_4 = \begin{bmatrix}
b^4_7= \{  1 & 11 & 21 & 24  \}\\
b^4_6= \{  2 & 12 & 15 & 25  \}\\
b^4_5= \{  3 & 13 & 16 & 26 & 29  \}\\
b^4_4= \{  4 & 14 & 17 & 27 & 30  \}\\
b^4_3= \{  5 & 8 & 18 & 28  \}\\
b^4_2= \{  6 & 9 & 19 & 22   \}\\
b^4_1= \{  7 & 10 & 20 & 23  \}\\
\end{bmatrix},
\end{equation*}
 
\begin{equation*}
P_5 = \begin{bmatrix}
b^5_7= \{      1 & 12 & 16 & 27   \}\\
b^5_6= \{ 2 & 13 & 17 & 28   \}\\
b^5_5= \{  3 & 14 & 18 & 22   \}\\
b^5_4= \{  4 & 8 & 19 & 23   \}\\
b^5_3= \{  5 & 9 & 20 & 24   \}\\
b^5_2= \{  6 & 10 & 21 & 25 & 29  \}\\
b^5_1= \{  7 & 11 & 15 & 26 & 30  \}\\
\end{bmatrix},
P_6 = \begin{bmatrix}
b^6_7= \{      1 & 13 & 18 & 23   \}\\
b^6_6= \{ 2 & 14 & 19 & 24 & 29  \}\\
b^6_5= \{ 3 & 8 & 20 & 25 & 30  \}\\
b^6_4= \{ 4 & 9 & 21 & 26   \}\\
b^6_3= \{ 5 & 10 & 15 & 27   \}\\
b^6_2= \{ 6 & 11 & 16 & 28  \}\\
b^6_1= \{ 7 & 12 & 17 & 22   \}\\
\end{bmatrix},
\end{equation*}
 
\begin{equation*}
P_7 = \begin{bmatrix}
b^7_7= \{1 & 14 & 20 & 26   \}\\
b^7_6= \{ 2 & 8 & 21 & 27  \}\\
b^7_5= \{ 3 & 9 & 15 & 28   \}\\
b^7_4= \{ 4 & 10 & 16 & 22   \}\\
b^7_3= \{ 5 & 11 & 17 & 23 & 29  \}\\
b^7_2= \{ 6 & 12 & 18 & 24 & 30  \} \\
b^7_1= \{ 7 & 13 & 19 & 25   \}\\
\end{bmatrix},
P_8 = \begin{bmatrix}
b^8_7= \{1 & 8 & 15 & 22 & 29  \}\\
b^8_6= \{ 2 & 9 & 16 & 23 & 30  \}\\
b^8_5= \{ 3 & 10 & 17 & 24   \}\\
b^8_4= \{ 4 & 11 & 18 & 25   \}\\
b^8_3= \{ 5 & 12 & 19 & 26  \}\\
b^8_2= \{ 6 & 13 & 20 & 27   \}\\
\bar{b}^8_1= \{ 7 & 14 & 21 & 28   \}\\
\end{bmatrix},
\end{equation*}
\end{enumerate}
\normalsize

Note that here all the blocks are not of the same size, but any two blocks from different parallel classes have at most 1 element in common, i.e., $\mu = 1$. Except for the blocks of parallel class $P_1$, the blocks of the remaining parallel classes have sizes in the set $\{ q-(e+f),q-(e+f)+1,\ldots ,q-e \}= \{4, 5\}$. Thus, we discard the Parallel class $P_1$. The remaining parallel classes form the required resolvable design, which we call RBD$(X,A)$.  Let us formalize the algorithm for the general case. Let $d= k \times s =(q-e)(q-f)$, with $0< f \leq e < q$. 
  
\begin{cons}
\label{con:1(q-e)(q-f)}
\normalfont
Let $q$ be a prime power, construct $(q^2, q, 1)$-ARBIBD. Call this design $(\bar{X}, \bar{A})$ with $\bar{X} = \{1, 2, \ldots, q^2\}$ and $|\bar{A}| = q(q+1)$ many blocks, each block is of constant size $q$. It will have $r = q+1$ many parallel classes, call them $\{\bar{P}_1,\bar{P}_2,\ldots,\bar{P}_{q+1}\}$, each parallel class having $q$ many blocks of constant size $q$. Between any two blocks from different parallel classes, exactly one element will be common, i.e., $|\bar{b}^l_i \cap \bar{b}^m_j| = 1, \forall \ l \neq m$.
\begin{enumerate}
\item  Given $e\geq f$, choose $e (\geq 0)$ many  blocks from $\bar{P}_1=\{\bar{b}^1_1, \bar{b}^1_2, \ldots, \bar{b}^1_h\}$. Let $S_1 = \bar{b}^1_1 \cup \bar{b}^1_2 \cup \ldots \cup \bar{b}^1_h $. Therefore, $|S_1| = e\times q$. 

\item From $\{\bar{b}^1_{e+1}, \bar{b}^1_{e+2}, \ldots, \bar{b}^1_{e+f}\}$ blocks of $\bar{P}_1$, choose any $(q-e)$ number of  elements from each of them. Let $S_2$ be the union of all these elements. Therefore, $|S_2| = f\times (q-e)$. Let $S = S_1 \cup S_2$.  

\item Remove the elements of set $S$ from RBD $(\bar{X},\bar{A})$ and remove parallel class $\bar{P}_1$ from $\bar{A}$. Call the resulting design RBD$(X,A)$.  
\end{enumerate} 
\end{cons}
 
We claim that RBD$(X,A)$ satisfies $|X| = (q-e)(q-f)$ and $A$ consists of $q+1$ many parallel classes having different block sizes, such that blocks from different parallel classes have at most one element in common, i.e., $|b^l_i \cap b^m_j| \leq 1$ for all $l\neq m$, implying $\mu=1$. We formalize this in the following lemma. 
 
\begin{lemma}
\label{lem:con1(q-e)(q-f)}
Let $d=(q-e)(q-f)$ for $f,e \in\mathbb{N}$ with $0<f \leq e \leq q$ and for some prime power $q$. Then one can construct an RBD$(X,A)$, with $|X|=d$ having block sizes from the set of integers $ \{(q-e-f), (q-e-f+1), \ldots, (q-e)\}$ with $\mu = 1$, and having $q$ many parallel classes.
\end{lemma}
\begin{proof}
Refer to Construction \ref{con:1(q-e)(q-f)} above. Here RBD$(\bar{X}, \bar{A})$ is an ARBIBD with $|\bar{X}|= q^2$, having a constant block size $q$. Note that any pair of blocks from different parallel classes have exactly one element in common, i.e., $ |\bar{b}^l_i \cap \bar{b}^m_j| = 1 \ \forall \ l \neq m$. The number of elements in the set $|S| = |S_1\cup S_2| = |S_1| + |S_2| = eq+ f(q-e) < q^2$, which is a proper subset of $\bar{X}$. These elements are removed from all the parallel classes of RBD$(\bar{X},\bar{A})$. Hence, the resulting design RBD$(X,A)$ is such that $|X| = q^2-eq-f(q-e) = (q-e)(q-f)= d$, having the same number of parallel classes as in $\bar{A}\setminus P_1$, which is $q$. The number of elements common between any two blocks from different parallel classes would be at most 1, i.e., $ |b^l_i \cap b^m_j| \leq 1, \forall \ l \neq m$.

To obtain the sizes of the blocks in RBD$(X, A)$, note that $S_1$ contains all elements from $e$ number of blocks of $\bar{P_1}$. Hence, $S_1$ would have at least $e$ elements in common with all the blocks of the remaining parallel classes. Thus, removal of the elements in $S_1$ from the parallel classes $\bar{P}_l, l \geq 2$ will remove at least $e$ elements from each block of $\bar{P_l}$, which implies $|\bar{b}^l_i \setminus S_1| = q-e$. Further, $S_2$ contains $q-e$ elements from $f$ many blocks of $\bar{P}_1$. Thus, the blocks in $\bar{P}_2, \bar{P}3, \ldots, \bar{P}{q+1}$ will have at most $f$ elements in common with $S_2$. Consequently, after the removal of all the elements in $S$ from the parallel class $\bar{P}_l$, the block size $|b^l_i|, \ l\geq 2$ will be at maximum $q-e$ and minimum $q-e-f$. However, the blocks in parallel class $P_1$ will be of sizes $q$ or $e$ and have a total of $(q-e)$ blocks, which we discard.
\end{proof}

Now, the RBD$(X, A)$ can be used to construct Approximate MUBs. Since the number of parallel classes in RBD$(X, A)$ is $q$, and as per Theorem \ref{thm:d}, we will obtain $q$ many $\beta$-AMUBs. Gain although $\mu =1$ here, the block sizes are not constant, hence $\Delta$ will consist of more than two elements, and thus it will not satisfy the criteria of APMUB, even though $\beta = 1+ \mathcal{O}(d^{-\lambda})$, where $\lambda= \frac{1}{2}$. For this situation we first define $\Delta_2$ as follows.

\begin{definition}
\label{def:delta2}
For $d=(q-e)(q-f)$, where $e\geq f$ and let $\theta_1 = \frac{1}{\sqrt{q-e-f}}, \theta_2 = \frac{1}{\sqrt{q-e-f+1}},\ldots , \theta_{f+1} = \frac{1}{\sqrt{q-e}}$, then $\Delta_2 = \{\theta_i \theta_j \} \cup \{0\}$ where  $i,j = 1,2,\ldots,(f+1)$. 
\end{definition}

Note that $|\Delta_2| = {f+1\choose 2} + (f+1) +1 = \frac{(f+1)(f+2)}{2}+1 = \mathcal{O}(f^2)$. Again, as for case of $\Delta_1$, here also the number of elements in $\Delta_2$ is also only dependent on the value of $f$ and is proportional to the square of it. We now state and prove the result on $\beta$-AMUB, using such RBD$(X,A)$ as follows

\begin{corollary}
\label{cor:(q-e)(q-f)}
Let $d=(q-e)(q-f) $ and $q$ be some prime-power, $0<f\leq e $ and $0 < (e+f) \leq \frac {3}{2} \sqrt{d}$ where $e,f \in \mathbb{N}$  then there exist at least $q$ many $\beta$-AMUBs, with $\Delta \subseteq \Delta_2$, $\beta = \frac{\sqrt{d}}{q-(e+f)}=1 + \frac{e+f}{2\sqrt{d}} +\mathcal{O}(d^{-1})\leq 2$  and $1- \frac{q-e}{d} \leq \epsilon \leq 1- \frac{1}{q}$ where $\epsilon$ denotes the sparsity.
\end{corollary}

\begin{proof}
Consider the RBD$(X,A)$ from Construction \ref{con:1(q-e)(q-f)}, where $|X|= (q-e)(q-f)$ with $f \leq e$. The block sizes of RBD$(X,A)$ are from the set $\{(q-e-f),(q-e-f+1),\ldots, (q-e) \}$ as given in Lemma \ref{lem:con1(q-e)(q-f)}. Using the Hadamard matrices of the order of block sizes, we get $q$ many Orthonormal Basis. Thus, the normalizing factors of Hadamard matrices of different order would be from set  $\left\{\frac{1}{\sqrt{q-(e+f)}},\frac{1}{\sqrt{q-(e+f)+1}},\ldots, \frac{1}{\sqrt{q-e}}\right\}$. Let's denote this set of factors by $S_\theta= \left\{\theta_1,\theta_2, \ldots, \theta_{f+1} \right\}$ where $\theta_1 = \frac{1}{\sqrt{q-(e+f)}}$, $\theta_2 = \frac{1}{\sqrt{q-(e+f)+1}}, \ldots \theta_{f+1} = \frac{1}{\sqrt{q-e+f}}$. There will be  $f+1$ elements in the set. Since $\mu = 1$ for the RBD$(X, A)$, we have $\left|\braket{\psi^l_i|\psi^m_j}\right| = \theta_i \theta_j $ or $0$, where $\ket{\psi^l_i}$ and $\ket{\psi^m_i}$ are the vectors from two different orthonormal bases constructed using a parallel class of RBD$(X, A)$, and $\theta_i, \theta_j \in S_\theta$. Thus, $\Delta \subseteq \Delta_2$.

Now, $\max \{ \left|\braket{\psi^l_i|\psi^m_j }\right| \} = \max \{ \theta_i \theta_j \} = \frac{1}{q-(e+f)}$ corresponding to the maximum value of $\theta_i = \theta_j = \frac{1}{\sqrt{q-(e+f)}}$. Thus, $\beta = \frac{\sqrt{d}}{q-(e+f)}$ hence if $\beta\leq c$, using $d= (q-e)(q-f)$ and solving for $q$ in terms of $d$, we get $q= d^\frac{1}{2} (1+ \frac{(e-f)^2}{4 d} )^{\frac{1}{2}} + \frac{e+f}{2} \Rightarrow efc + (e+f)\sqrt{d} \leq \frac{c^2-1}{c} d$, which for $c= 2$ we get $\frac{2ef}{\sqrt{d}}+ (e+f) \leq \frac{3}{2} \sqrt{d}$ hence $0 < (e+f) < \frac {3}{2} \sqrt{d}$. Note that for this condition, we have  $\frac{e+f}{2 \sqrt{d}} < \frac{3}{4} < 1$, thus the series expansion gives $\beta =1 + \frac{e+f}{2\sqrt{d}} +\mathcal{O}(d^{-1})\leq 2$.

To estimate sparsity, refer to the Lemma \ref{lem:sparsity}, since there are $q$ many blocks in each parallel class, the $\epsilon$ is bounded above by $1-\frac{1}{q}$, and for lower bound, the maximum block size is $q-e$, thus Corollary \ref{lem:sparsity} gives $\epsilon$ is bounded below by $1- \frac{q-e}{d}$, hence 
$$1- \frac{q-e}{d} \leq \epsilon \leq 1- \frac{1}{q}.$$
\end{proof}

For example, for the case of RBD$(X,A)$ with $d= 5\times 6 = (7-2)(7-1)=30 $, implying $e=2$ and $f= 1$ shown in the example above, when used for constructing the $\beta$-AMUBs, we get $\beta =\frac{ \sqrt{30}}{4}= 1.37$ and the set  $\Delta = \{ \frac{\beta_i}{\sqrt{d}} \}$, where  $\beta_i= \{0, 0.76, 0.93, 1.04, 1.09, 1.20, 1.37\}$ 

The Corollary \ref{cor:(q-e)(q+f)} and Corollary \ref{cor:(q-e)(q-f)} together  imply that for every composite $d = k\times s$ such that $|s-k| \leq \sqrt{d}$, then we can always construct $\mathcal{O}(\sqrt{d})$ many APMUB. We formally state and prove following 

\begin{lemma}
For any composite number, $d = k\times s$, $k\leq s$ such that  $ 2\delta = s- k\leq \sqrt{d} $. Then there exist $\mathcal{O}(\sqrt{d}~)$ many $\beta$-AMUB where $\beta \leq  1 + \frac{\delta }{\sqrt{d}} + \mathcal{O}(d^{-\lambda}) \leq 2 $ where $\lambda = \frac{1-\theta}{2}= 0.2375 $ and $  \epsilon = 1- \mathcal{O}(d^{-\frac{1}{2}}) $.
\end{lemma}

\begin{proof}
Consider the prime number $p$ nearest to $u = \frac{k+s}{2}$ but greater than $u$. From result on gap in prime \cite{baker2001difference}, there will exist a prime number $p$ in   the interval $[u, u + \mathcal{O}(u^\theta)]$ where $\theta = 0.525$. Now let $v = p - u  $, hence $v \leq \mathcal{O}(u^\theta)$. Consequently,  $s = p - v + \delta$ and $k = p - v - \delta$ where $2 \delta = s- k$. Now, if $v \leq \delta$, then it becomes the case of $d= (q-e)(q+f)$ where $e= v + \delta$ and $f= v - \zeta$, in which case we get $p+1$ many $\beta$-AMUB with $\beta = 1+ \frac{\delta}{\sqrt{d}} +\mathcal{O}(d^{-1}) \leq 2$. On the other hand if $v \geq \delta$, then it becomes the case of $d=(q-e)(q-f)$ where $q= p$, $e= v + \delta$, and $f= v - \delta$ in which case we get $p$ many $\beta$-AMUB with $\beta = 1+ \frac{e+f}{2\sqrt{d}} +\mathcal{O}(d^{-1})$. And as shown in Lemma \ref{lem:e+f}, we have $\frac{e+f}{2}= \delta + f = \delta + \mathcal{O}(d^{\frac{\theta}{2}})$. Hence $\beta =   1 + \frac{\delta }{\sqrt{d}} + \mathcal{O}(d^{-\lambda}) \leq 2 $ where $\lambda = \frac{1-\theta}{2}= 0.2375$.
  
Thus, considering both the cases together, we get $p+1$ or $p$ many $\beta$-AMUB with $\beta \leq   1 + \frac{\delta }{\sqrt{d}} + \mathcal{O}(d^{-\lambda}) \leq 2 $. Since for $d = (p-e)(p+f)$ or $d = (p-e)(p-f)$, with $e\geq f$, we have $d = p^2-(e\pm f)p \mp ef \Rightarrow p = \mathcal{O}(\sqrt{d}) $, hence number of $\beta$-AMUB is $\mathcal{O}(\sqrt{d})$. The sparsity $\epsilon$ is bounded from below by $ 1- \frac{s}{d}$ and bounded from above by $1- \frac{1}{q}$  hence  $\epsilon =1- \mathcal{O}(d^{-\frac{1}{2}}) $.
\end{proof}

Note that, the number of $\beta$-AMUBs will always be $\geq \lfloor d^\frac{1}{2}\rfloor$.  It can also be verified that, if we assume the validity of Cram\'er Conjecture \cite{cramer1936order}, the above result will hold, but now the asymptotic series for $\beta$ for the case of $d=(q-e)(q-f)$ will be $\beta = 1 + \frac{\delta }{\sqrt{d}} + \mathcal{O}((\log^2 d) d^{-\frac{1}{2} }) \leq 2$. Thus if $\delta$ is bounded then we have $\beta = 1 +  \mathcal{O}( d^{-\frac{1}{2} })$.

Let us now focus on using RBD, having constant block sizes. Such RBD has the advantage of using a single Hadamard matrix for the entire construction. Hence, RBD$(X, A)$ are more amenable for ARMUB construction. Whereas RBD$(X, A)$ with non-constant block size, intersection number $\mu$ is generally small and easy to construct, but then different order Hadamard matrices are required for constructing AMUB. And since real Hadamard matrix exist for order $2$ or multiple of $4$, thus getting ARMUB using such RBD having variable block sizes may not be possible. But our experience has shown that in general, constructing  RBD$(X, A)$, having a large number of parallel classes, and having constant block size for all the parallel classes such that intersection number $\mu$ remains small and bounded are difficult to achieve. In this direction, we present a few constructions of AMUB/ARMUB through RBD having constant block size.

\subsection{Construction of $\beta$-AMUB through RBD having constant block size}

In \cite[Theorem 2]{ak23}, it is shown that for any composite dimension $d = k \times s,\,  k\leq s$ such that $\sqrt{\frac{s}{k}} < 2 $, one can construct at least $N(s) +1$ APMUBs. Here, the key idea is the use of MOLS(s) for the construction of RBD$(X, A)$, having $N(s)+1$ many parallel classes with $\mu =1 $. This also enables us to construct ARMUB using the real Hadamard matrix of order $k$ if it exists. However, if only the Hadamard matrix of order $s$ exists and not $k$,  then to construct ARMUBs, the RBD$(X, A)$ should have block size  $s$. Nevertheless, from \cite[Lemma 2]{ak23}, we know that in such a situation, $\mu \geq 2$. Hence, the minimum value of $\mu$ would be 2 in such a situation. Also, if $N(k) > N(s)$ and we wish to use the MOLS$(s)$ to get RBD with more number of parallel classes, then in such situation we would like to have  RBD, with a block of size $s$.

Our following construction achieves the minimal possible value of $\mu = 2$ when block size $0< k < 2s$. We express $k = s+ f$, $0< f\leq s$.  Let us demonstrate the construction by explicitly constructing RBD$(X,A)$ with $|X|= 5(5+3)=10$, here $q=5$ and $f=3$. As in  \cite[Section 5.1]{ak23}, $4$- MOLS($5$) was used to get RBD$(\bar{X}, \bar{A})$ having $6$ parallel classes. Following the steps of \cite[Construction 3]{ak23} for $d=(q-e)q$, where $e= q-f$, we will use this, to construct RBD$(\widetilde{X},\widetilde{A})$  with $|\widetilde{X}| = (5-2)5=15$. For this let, $\widetilde{X} = \{26,27,28,29,30,31,32,33,34,35,36,37,38,39,40 \}$. Now, we will combine this design of RBD$(\bar{X}, \bar{A} )$. We will get RBD(X, A), with elements numbered 1 to 40. Explicitly RBD$(\widetilde{X},\widetilde{A})$  with $|\widetilde{X}| = 5(5-2)=15 $ is as follows.

\small
\begin{equation*}
\widetilde{P}_1 = \begin{bmatrix}
     \widetilde{b}^1_5=  \{ 26& 32   &38   \}  \\
     \widetilde{b}^1_4=  \{ 27& 33   &39   \}  \\
     \widetilde{b}^1_3=  \{ 28& 34   &40   \} \\
     \widetilde{b}^1_2=  \{ 29& 35 &36   \} \\
     \widetilde{b}^1_1= \{ 30& 31   &37    \}\\
 \end{bmatrix},
\widetilde{P}_2 = \begin{bmatrix}
     \widetilde{b}^2_5= \{26&33  &40   \}\\
     \widetilde{b}^2_4= \{ 27&34  &36   \} \\
     \widetilde{b}^2_3= \{28&35 &37 \}\\
     \widetilde{b}^2_2= \{ 29 &31  &38  \}\\
     \widetilde{b}^2_1= \{30&32   &39 \}\\
 \end{bmatrix},
 \widetilde{P}_3 = \begin{bmatrix}
    \acute{ b}^3_5= \{ 26&34   &37  \}\\
    \acute{ b}^3_4= \{ 27&35 &38  \}  \\
    \widetilde{b}^3_3= \{  28&31   &39   \}\\
    \widetilde{b}^3_2= \{  29&32   & 40\} \\
    \widetilde{b}^3_1= \{  30&33   &36  \} \\
\end{bmatrix},
\end{equation*}
\begin{equation*}
\widetilde{P}_4 = \begin{bmatrix}
     \widetilde{b}^4_5=\{ 26&35 &39\} \\
     \widetilde{b}^4_4=\{ 27&31  &40 \}   \\
     \widetilde{b}^4_3=\{ 28&32  &36\}  \\
     \widetilde{b}^4_2=\{ 29 &33 &37  \} \\
     \widetilde{b}^4_1=\{  30&34 &38  \} \\
\end{bmatrix},
\widetilde{P}_5 = \begin{bmatrix}
   \widetilde{b}^5_5= \{26&31&36\} \\
   \widetilde{b}^5_4=\{  27&32 &37   \} \\
   \widetilde{b}^5_3=\{  28&33  &38 \} \\
   \widetilde{b}^5_2= \{ 29 &34  &39 \} \\
   \widetilde{b}^5_1=\{    30&35 &40 \}  \\
 \end{bmatrix},
\end{equation*}
\normalsize

Note that, any two blocks from different parallel class of RBD$(\widetilde{X},\widetilde{A})$ has at most one point in common. Now taking corresponding block wise union of RBD$(\widetilde{X},\widetilde{A})$ with RBD$(\bar{X},\bar{A})$, i.e., $b^i_j= \bar{b}^i_j \cup \widetilde{b}^i_j$. Ignore any one of the parallel class of RBD$(\bar{X},\bar{A})$. The resulting RBD$(X,A)$ is given as follows.

\small
\begin{equation*}
P_1 = \begin{bmatrix}
     b^1_5=  \{ 1& 7   &13  &19  &25& 26& 32   &38\}  \\
     b^1_4=  \{ 2& 8   &14  &20  &21 &27& 33   &39 \}  \\
     b^1_3=  \{ 3& 9   &15  &16  &22 &28& 34   &40\} \\
     b^1_2=  \{ 4& 10 &11  &17   &23 &29& 35 &36\} \\
     b^1_1= \{ 5& 6   &12  &18  &24 & 30& 31   &37\}\\
 \end{bmatrix},
P_2 = \begin{bmatrix}
    \check{ b}^2_5= \{1&8  &15  &17  &24 &26&33  &40\}\\
     b^2_4= \{ 2&9  &11  &18  &25 &27&34  &36 \} \\
     b^2_3= \{3&10 &12 &19  &21& 28&35 &37\}\\
     b^2_2= \{ 4 &6  &13 &20  &22 &29 &31  &38\}\\
     b^2_1= \{5&7   &14 &16  &23&30&32   &39\}\\
 \end{bmatrix},
\end{equation*}
\begin{equation*}
P_3 = \begin{bmatrix}
     b^3_5= \{ 1&9   &12 &20 &23 &26&34   &37 \}\\
     b^3_4= \{ 2&10 &13 & 16 &24 &27&35 &38\}  \\
    b^3_3= \{  3&6   &14 &17 &25 &28&31   &39 \}\\
    b^3_2= \{  4&7   &15 &18 &21& 29&32   & 40\} \\
    b^3_1= \{  5&8   &11 &19  &22& 30&33   &36 \} \\
 \end{bmatrix},
P_4 = \begin{bmatrix}
     b^4_5=\{ 1&10 &14&18&22 &26&35 &39 \} \\
     b^4_4=\{ 2&6  &15 &19&23& 27&31  &40\}   \\
     b^4_3=\{ 3&7  &11&20&24& 28&32  &36 \}  \\
     b^4_2=\{ 4 &8 &12 &16 &25 &29 &33 &37\} \\
     b^4_1=\{  5&9 &13 &17 &21&30&34 &38 \} \\
 \end{bmatrix},
\end{equation*}
\begin{equation*}
P_5 = \begin{bmatrix}
   b^5_5= \{1&6&11&16&21& 26&31&36 \} \\
  \check{ b}^5_4=\{  2&7 &12 &17 &22 & 27&32 &37 \} \\
   b^5_3=\{  3&8  &13&18 &23&28&33  &38 \} \\
   b^5_2= \{ 4 &9  &14 &19 &24 &29 &34  &39\} \\
   b^5_1=\{    5&10 &15 &20 &25 &30&35 &40 \}  \\
\end{bmatrix}.
\end{equation*}
\normalsize

It can be seen that in the above RBD$(X,A)$, any two blocks from different parallel class have either one or two points in common. $X$ consist of points from $1$ to $40$, and $A =\{P_1 \cup P_2 \cup P_3 \cup P_4 \cup P_5\}$.
The technique is more formally explained for the general case in the following Construction.

\begin{cons}
\label{cons:(s+f)s}
\normalfont
Let $d= k\times s = (s+f)s$, with $0< f\leq s$
\begin{enumerate}

\item Using \cite[Construction 1]{ak23}, construct an $RBD(\bar{X},\bar{A})$ using $w$-MOLS($s$).  The resulting RBD$(\bar{X},\bar{A})$ has $\bar{X} = \{1,2,\ldots,s^2\}$ and $|A| = s(N(s)+2)$ many blocks. It will have $r = N(s)+2$ many parallel classes, namely $\left\{\bar{P}_1,\bar{P}_2,\ldots,\bar{P}_{N(s)}, \bar{P}_0,\bar{P}_\infty \right\}$, each having $s$ many blocks of constant size $s$. The blocks of the parallel class $\bar{P}_l$ are denoted by $\bar{b}^l_i, i = 1,2,\ldots s $. Between any two blocks from different parallel class, have exactly one point in common i.e., $|\bar{b}^l_i \cap \bar{b}^m_j| = 1\,\, \forall\hspace{1mm}l \neq m$.

\item Pick any parallel class, say $\bar{P}_{1}$. Remove $(s-f)$ many blocks from it and let $S = \{ b^1_1 \cup b^1_2\cup \ldots \cup b^1_{(s-f)}\}$.

\item Remove all the points of $S$ from $\bar{X}$ i.e., $\bar{X} \backslash S$ and also remove the points of $S$ from all the blocks of parallel classes  $\{\bar{P}_2, \bar{P}_3, \ldots,\bar{P}_{s+1} \}$. Let the resulting parallel classes be called as $\{ \widetilde{P}_2, \widetilde{P}_3,\ldots,\widetilde{P}_{q+1}\} $ i.e., $\widetilde{P}_i = \bar{P}_i \backslash S$

\item Discard the parallel class $\bar{P}_1$.

\item Construct another $RBD(X,A)$ having elements from $(\bar{X}, \bar{A})$. Then $ |X| = s^2$ and $|A| = s(N(s)+2)$ blocks with $r= N(s)+2$ many parallel classes $\{P_1,P_2,\ldots,P_{r}\}$.

\item Discard any one parallel class from $A_1$, say $P_1$.

\item Form a new design $(X,A)$ where $X = \tilde{X} \cup X$, $A=\{P_2,P_3,\ldots, P_r \} $ where $P_l = P_l + \widetilde{P}_l$. Then $(X,A)$ is the required RBD.
\end{enumerate}
\end{cons}

We claim that the above design $(X, A)$ is an RBD such that $|X| = s(s+f)$ and $A$ consists of $N(s)+1$ many parallel classes, $A = \{P_2,\ldots, P_r \}$, each parallel class have $s$ many blocks, $P_l = \{ b^l_1, b^l_2,\ldots,b^l_s\}, l = 2,3,\ldots,r$, each of size $(s+f)$ i.e., $|b^l_i | = (s+f) \,\, \forall\hspace{1mm}l,i$ such that blocks from different parallel classes have at most two points in common, i.e., $ 1 \leq  |b^l_i \cap b^m_j | \leq 2\,\, \forall \, i \neq j $. We formalize this using the following lemma.

\begin{lemma}
\label{lem:(s+f)s}
Let $d=(s+f)s$ with $0< f\leq s$, then one can construct an RBD$(X,A)$, with $|X|= d$ having constant block size of $(s+f)$ with $\mu = 2$ and having $N(s)+1$ many parallel classes, where $N(s)$ is the number of $MOLS(s)$.
\end{lemma}
\begin{proof}
Refer to Construction \ref{cons:(s+f)s}. Since any pair of blocks from different parallel classes is of size $s$ and has exactly one point in common in RBD$(\bar{X},\bar{A})$, i.e., $|\bar{b}^l_i \cap \bar{b}^m_j| = 1, \,\forall , l \neq m$, removing the elements of $S = \left\{\bar{b}^1_1 \cup \bar{b}^1_2\cup \ldots \cup \bar{b}^1_{s-f} \right\}$ from the entire design will remove exactly $(s-f)$ elements from each block $\bar{b}^l_t,\, l\neq 1$. Hence, the blocks $\tilde{b}^l_t= \bar{b}^l_t \backslash S$ will be of constant size $|\tilde{b}^l_i| = f$ and $ |\tilde{b}^l_i \cap \tilde{b}^m_j| \leq 1,, \forall, l \neq m$.

On the other hand, in an RBD$(X, A)$, any pair of blocks from different parallel classes is of size $s$ and has exactly one point in common, i.e., $ |b^l_i \cap b^m_j| = 1,\, \forall , l \neq m $. Since the design $(X, A)$, where $X = \tilde{X} \cup \bar{X}$, $A=\left\{P_2, P_3,\ldots, P_r\right\}$, and $P_l = P_l + \widetilde{P}_l$, is a direct union of the blocks for $\bar{A}$ and $\tilde{A}$, it will have either one point or two points in common between blocks of different parallel classes.
\end{proof}
 
Now using such RBD, we can construct $\beta$-AMUBs with the following characteristics. 

\begin{theorem}
\label{thm:d=s(s+f)}
If $d=s(s+f) $, with $s, f \in \mathbb{N}$ and $f \leq s$, then one can construct $N(s)+1$ many approximate MUBs with $\beta = 2 \sqrt{\frac{s}{s+f}} = 2 - \frac{ f}{\sqrt{d}}+\mathcal{O}(d^{-1}) \leq 2$ and sparsity $\epsilon = 1-\frac{1}{s}$. If there exist a real Hadamard matrix of order $(s+f)$, then one can construct $N(s)+1$ many approximate real MUBs (ARMUBs) with the same $\beta$ and $\epsilon$. Furthermore, $\Delta = \left\{ 0, \frac{\beta_o}{\sqrt{d}}, \frac{2\beta_o}{\sqrt{d}} \right\}$ where $\frac{1}{\sqrt{2}} \leq \beta_o= 1 - \frac{ f}{2 \sqrt{d}}+\mathcal{O}(d^{-1})  < 1$.
\end{theorem}
\begin{proof}
Following the Construction \ref{cons:(s+f)s}, we construct an RBD$(X,A)$ with $|X|= d=(s+f)s$. The block size is $(s+f)$, and the number of parallel classes is $N(s)+1$. Since the intersection number $\mu=2$, we have $\beta= \frac{2\sqrt{d}}{s+f}=2 \sqrt{\frac{s}{s+f}} < 2$. The result follows from the construction of $\beta$-AMUBs in \cite{ak22} and Theorem \ref{thm:d=ks}. The minimum possible $\beta$ in this situation is when $f= s$, for which $\beta \leq \sqrt{2}$. The asymptotic variation of the parameters in terms of $d$ is given by $\beta=2-\frac{f}{\sqrt{d}}+ \mathcal{O}(d^{-1})$. However, here, $\beta$ converges to 2. The sparsity is given by $\epsilon = 1- \frac{s+f}{d} = 1 - \frac{1}{s}$. Using the real Hadamard matrix of order $(s+f)$, we obtain the ARMUBs with the same $\beta$ and $\epsilon$. However, the set $\Delta$, which contains all the different values of the absolute value of dot product $|\braket{\psi^l_i|\psi^m_j}| $ of vectors $\ket{\psi^l_i}$ and $\ket{\psi^m_j}$ from different bases, is restricted to set $\{ 0,\frac{1}{s+f}, \frac{2}{s+f} \}$. Hence, $\Delta = \left\{0, \frac{\beta_o}{\sqrt{d}}, \frac{ 2\beta_o}{\sqrt{d}}\right\}$ where $\beta_o =\sqrt{ \frac{s}{s+f}} $ hence $\frac{1}{\sqrt{2}} \leq \beta_o= 1 - \frac{ f}{2 \sqrt{d}}+\mathcal{O}(d^{-1})  < 1$, where the lower bound correspond to the situation when $f= s$.
\end{proof}

When $s = q$, some power of a prime number, there is well-known method of construction of affine resolvable $(q^2,q,1)$-BIBD which is an RBD. In this regard, we have the following immediate corollary.
 
\begin{corollary}
\label{cor:q(q+f)}
If $d=q(q+f) $, where $q$ is some power of a prime and $q, f \in \mathbb{N}$ such that $f \leq q$, then we can construct $q$ many AMUB with $\beta =2 \sqrt{\frac{q}{q+f}}= 2 - \frac{ f}{\sqrt{d}}+ \mathcal{O}(d^{-1})$. Moreover, if there exist a real Hadamard matrix of order $(q+f)$, one can construct $q$ many approximate real MUBs (ARMUBs) with the same parameters. 
\end{corollary}

The proof of this corollary follows directly from the fact that $N(q) = q-1$. For example, the RBD$(X,A)$ constructed above having  $5$ parallel classes, can be converted into $5$ orthonormal bases, which gives $\beta = 2\sqrt{\frac{5}{8}} =  1.58 < 2$ and $\epsilon = 1 - \frac{1}{5}= 0.8$. 

The corollary above is a generalization of the result for $d= q(q+1)$ \cite[Theorem 4]{ak22} to $q(q+f)$, having a similar parameter, $\beta = 2 - \mathcal{O}(d)$. In fact, the result of \cite[Theorem 4]{ak22} is a particular case of our present result, where the block size is larger than the number of blocks, with $e= 0, f= 1$, and $\mu = 2$. There are $q+1$ many parallel classes in an RBD, each having a constant block size of $(q+1)$. In that case, $\beta = 2 \sqrt{\frac{q}{q+1}} = 2 - \mathcal{O}(\frac{1}{\sqrt{d}})$, which is the same as \cite[Theorem 4]{ak22}. Once again, we would like to point out that these are not APMUBs but provide results of the same quality in terms of {absolute inner product} values as \cite{ak22}, but over a {larger class}. In order to obtain APMUBs, we must have $\mu=1$.

The above Theorem \ref{thm:d=s(s+f)} and  \cite[Theorem 2]{ak23}, together give the following important corollary, which enables us to construct $\beta$-ARMUBs, such that $\beta< 2$, for every $d= k\times s, \, k\leq s$ such that $s-k < \sqrt{d}$ and there exist a real Hadamard matrix of order $k$ or $s$. The quality of the constructed $\beta$-ARMUB depends on the factors of $d$ and $|s-k|$.

\begin{corollary}
\label{cor:d=ks}
Let $d= k\times s$, with $|s-k|< \sqrt{d}$. If a real Hadamard matrix of order $k$ exists, then one can construct $N(s)+1$ many $\beta$-ARMUB, with sparsity $\epsilon = 1- \frac{1}{k}$. If $k< s$ then  $\beta = \sqrt{\frac{s}{k}} = 1 + \frac{\delta}{\sqrt{d}} + \mathcal{O}(\frac{\delta^2}{d} ) < 2$, and if $k> s$ then $\beta =  2 \sqrt{\frac{s}{k}} = 2 - \frac{\delta}{\sqrt{d}} +  \mathcal{O}(\frac{\delta^2}{d} )< 2$, and, if $k=s$, then $\beta = 1$  where $2\delta =|s-k|$.
\end{corollary}
\begin{proof}
When we have a Hadamard matrix of order $k$, and if $k< s$ then we employ the construction corresponding to $d= (s-e)s$, \cite[Construction 3]{ak23}, with $k =s-e$, which will result into  $N(s)+1$ many $\beta$-AMUB, with $\beta = \sqrt{\frac{s}{k}}$ and $\epsilon = 1- \frac{1}{s}$ as stated in \cite[Theorem 2]{ak23}. On the other hand when $k>s$ then we will employee above Construction \ref{cons:(s+f)s} ~for $d= s(s+f)$  to construct $N(s)+1 $ many $\beta$-AMUB with $\beta = 2 \sqrt{\frac{s}{s+f}}  = 2 - \frac{\delta}{\sqrt{d}} +  \mathcal{O}(\frac{\delta^2}{d} )< 2$.
\end{proof}

The quality of the constructed AMUB for $d = k \times s$ depends on $\delta = |s - k| < \sqrt{d}$, and the smaller the $\delta$, the closer $\beta$ becomes to $1$. However, note that the number of AMUB is only of the order of $N(k)$ or $N(s)$, which is generally small. Nevertheless note that $N(w) \rightarrow \infty $ as $w\rightarrow \infty$ whereas the number or real MUBs for most of the non square dimension is either 2 or 3.

\section{Discussion and comparison with existing results on AMUB} 
\label{sec:5}

We have shown that for a composite $d= k\times s$, if $|s-k| <\sqrt{d}$, then RBD can be used to constructed $\geq \sqrt{d}$ many very sparse  $\beta$-AMUB, with $\beta \leq 2$ for all such composite $d$. This is to be compared with the fact that corresponding to such composite $ d= k \times s = p_{n_1}^{n_1}\, p_{n_2}^{n_2}\, \ldots p_{n_m}^{n_m} $, number of MUB possible is $p_{n_r}^{n_r}+1$ where $p_{n_r}^{n_r} $ is min$\{ p_{n_1}^{n_1}, p_{n_2}^{n_2},\ldots p_{n_m}^{n_m}\}$. Thus  number of  $\beta$-AMUB will always be greater than MUBs for such composite $d$.  

In order to constructing AMUB, for a such  composite $d $, express $d= (q-e)(q+f)$ or $(q-e)(q-f)$, where  $q\geq \frac{|s+k|}{2}$ is some prime-power closest to $\frac{|s+k|}{2}$ but greater than it. Then construct RBD, whose block sizes are $\mathcal{O}(\sqrt{d}$, with $\mu =1$. The most important parameter which control the quality of AMUB, measured by closeness of $\beta$ to 1, is $\frac{|s-k|}{2}$. The order of set $\Delta$, which consist of different possible values of $|\braket{\psi^l_i|\psi^m_j}|$, where $\ket{\psi^l_i}$ and $\ket{\psi^m_j}$ are basis vectors from different bases  is $\mathcal{O}(f^2)$.  hence for small $f$, hence we get only a few different values of $|\braket{\psi^l_i|\psi^m_j}|$. And for the case of $d= (q-e)(q+f)$, the $0 \leq f  \leq \delta $ and for the case of $d= (q-e)(q-f)$ the $0 \leq f = \mathcal{O}(d^{\frac{\theta}{2} }) $. Thus smaller the $\delta$, the $\beta$ will be closer to 1 and $|\Delta |$ will be small. And when $\delta = 0 $ i.e., $d= q^2$, we get $q+1$ many  MUBs.  For example for $d= 6\times 10$ we can construct 10 $\beta$-AMUB with $\beta =1.29 $ where as for same $d$ we have only 4 MUBs. And for $d= 6\times 7$ we can construct 8 $\beta$-AMUB with $\beta =1.08 $ where as for same $d$ we have only 3 MUBs. Note that smaller the $\delta$, closer is the $\beta $ to 1

The RBD having constant block size is particularly useful for constructing $\beta$-ARMUB. We have shown that for a composite $d=k\times s$ with $|s-k| < \sqrt{d}$, such that a real Hadamard matrix of order $k$ or $s$ is available, then we can construct $N(s)+1$ or $N(k)+1$ many $\beta$-ARMUB with $\beta < 2$ respectively. For example for $d= 4\times 7 $ we can construct 7 $\beta$-ARMUB with $\beta =1.32 $ where as only 2 real MUB exist \cite[Table 1]{MUB2} in this case. Consider  for $d= 7\times 12 $ we can construct 7 $\beta$-ARMUB with $\beta =1.527 $ where as only 2 real MUB exist in this case as well \cite[Table 1]{MUB2}.

Here we generalize the result for $d= q(q+1)$ of~\cite{ak22} to $q(q+f)$ in Corollary \ref{cor:q(q+f)} having a similar form of $\beta = 2 - \mathcal{O}(d^{-\frac{1}{2}})$. We also improve the number of $\beta$-AMUBs for the case of $(q-e)(q+f)$ where previously only $\lfloor \frac{q-e}{f} \rfloor +1$ many $\beta$-AMUBs could be constructed. However, now $q$ many $\beta$-AMUBs could be constructed with similar $\beta = 1 + \mathcal{O}(d^{-\frac{1}{2}})$. On the other hand, $|\Delta|$ is now increased from 2 to $\mathcal{O}(f^2)$ defined in Definition \ref{def:delta1}. In fact, we generalized the case for $d = (q - e)(q + f), 0 < f \leq f$ for the construction of $\beta$-AMUBs to include the case for $d = (q - e)(q \pm f)$. Thus, for situation like $d= 9\times 10$ or $d= 13\times 15$ etc., we cannot construct APMUB, but can construct $\beta$-AMUB by expressing these $d$ in the form of $(q-e)(q-f)$ for suitable $e$ and $f$.  However, this generalization comes at the expense of increasing $|\Delta|$ to $|\Delta_1|$ (Definition \ref{def:delta1}) or to $|\Delta_2|$ (Definition \ref{def:delta2}), as opposed to the previous scenario of APMUB where $|\Delta| = 2$. 

We make the following observation applicable to various construction of $\beta$-AMUB here, highlighting common characteristics of AMUBs constructed using RBD$(X, A)$.

\begin{enumerate}
 
\item One of the critical parameters of RBD$(X, A)$ is the intersection number, denoted as $\mu$, which is the maximum number of elements common between any pair of blocks from different parallel classes.  Note that, as each parallel class contains all the points of $X$, a block in a parallel class is bound to have at least one point in common with some block of a different parallel class, and therefore $\mu \geq 1$. A lower value of $\mu$ is desirable for a lower value of $\beta$.

\item In general, using RBD$(X,A)$, with $|X|= d = k\times s$ a composite number, we can obtain $\beta$-AMUB (ARMUB) if $\delta = \frac{|s-k|}{2}$ is small. In fact, the smaller the value of $\delta$, the better the quality of the AMUB. For large values, we get poor-quality AMUB (ARMUB).
.
 
\item The other critical parameter is the block size of RBD$(X, A)$. The block sizes should be around $\sqrt{d}$ to obtain good quality AMUB (ARMUB). Closer the block size to $\sqrt{d}$, closer the value of $\beta$ to $1$.
 
\item In general, for composite $d$ with small $\delta$ and a resolvable design having a block size $\mathcal{O}(\sqrt{d})$ with small $\mu$, the $\beta$ of the constructed AMUB is of the form $\beta = \mu(1 \pm \frac{\delta}{\sqrt{d}} +\mathcal{O}(d^{-1}))$. In such a situation, generally, we get $\mathcal{O}(\sqrt{d})$ many APMUB.
 
\item For ARMUB, a Hadamard matrix corresponding to the block sizes of the parallel class should exist. Thus, RBD with block sizes equal to the order of some Hadamard Matrix is sufficient to construct ARMUB. However, since the real Hadamard matrix exists only for order $2$ or orders multiple of $4$, it is easier to construct ARMUB with RBD having a constant block size, a multiple of $4$. 
  
\item Sparsity $(\epsilon)$ of the constructed AMUB depends on the block sizes. The larger the block size, the smaller the sparsity, and vice versa.
 
\item The set of all the different absolute values of the dot product of basis vectors of AMUB, denoted by $|\Delta|$, is dependent on the number of different block sizes of RBD. The more different sizes of blocks in RBD, the larger the size of $|\Delta|$. Hence, RBD having a constant block size is desirable to get a smaller size of the set $|\Delta|$. Furthermore, $|\Delta|$ is also dependent on the value of $\mu$. The larger the value, the larger the size of $|\Delta|$. Hence, a smaller intersection number in RBD is also desirable to get a smaller size of the set $|\Delta|$.
\end{enumerate} 

The best result, applicable for most of the dimensions, for $\beta$-AMUB is based on the elliptic curve construction, \cite[Theorem 2]{AMUB2} where the construction gave $p^{m-1}$, $m\geq 2$ where $p$ is a prime such that $\sqrt{n} -1 \leq \sqrt{p} \leq \sqrt{n} +1.$
$$|\braket{\psi^l_i|\psi^m_j}| \leq \frac{2m + \mathcal{O}(d^{-\frac{1}{2} })}{\sqrt{d}} = \mathcal{O}(n^{-\frac{1}{2}}) \Rightarrow \beta = 2m + \mathcal{O}(d^{-\frac{1}{2} }).$$
Here the smallest value of $\beta = 4 +\mathcal{O}(d^{-\frac{1}{2} }) > 4$,  corresponding to $m=1$. However, here we could provide a construction where $\beta \leq 2$. Thus the $\beta$ is closer to one in all our construction than this. On the other hand, we obtain only $\mathcal{O}(\sqrt{d})$ many AMUBs, whereas \cite[Theorem 2]{AMUB2} can provide $\mathcal{O}(d)$ many AMUB. The other generic construction of AMUB applicable for all $d$ is of Klappenecker et al. \cite{AMUB1} where they gave construction of AMUB that has $\beta = \mathcal{O}(d^{\frac{1}{4}})$  \cite[Theorem 11]{AMUB1} or the construction of AMUB based on the finite field \cite[Theorem 1]{AMUB2} where $\beta = \mathcal{O}(\sqrt{\log d})$. Thus in all the known construction of AMUB for a generic $d$, the $\beta$ constructed using RBD is much closer to 1 than the other known construction. In fact here $\beta \rightarrow 1$ for larger $d$ where as for other it blows up without any bound or tends to a larger values. 

In case of certain specific kinds of $d$, as per our survey, only for $d = q - 1$, where $q$ is some power of prime, there are $d$ or $d+1$ AMUBs \cite{AMUB1,AMUB-MixedCharacterSum} where $\beta = 1+\mathcal{O}(d^{-\lambda})$ for $\lambda > 0$. The other known case of the $\beta$ of this form, for $d = q(q-1)$, the number of AMUBs is $\mathcal{O}(\sqrt{d})$, and for $d = \phi(n)$, the number of AMUBs is equal to the smallest prime division of $n$, which is always less than $\sqrt{n}$ when $n$ is not a prime number \cite{AMUB-FrobeniusRing}. At the same time, we have shown that for all composite $d = k \times s$ when $|s-k| < d^\frac{1}{2}$ then we will get more than $\sqrt{d}$ many AMUB with $\beta = 1+ \mathcal{O}(d^{-\lambda})$ for $\lambda > 0$. Thus, we can effectively construct such AMUBs for a large set of integer $d$. Further, we are also able to construct ARMUB with $\beta = 1+ \mathcal{O}(d^{-\lambda})$ or with $\beta = 1- \mathcal{O}(d^{-\lambda}) < 2$ for such $d$ whenever real Hadamard matrix of order $k$ or order $s$ is available. Moreover, all these AMUBs are very sparse where in general the sparsity $\epsilon = 1- \mathcal{O}(d^{-\frac{1}{2}})$ for both complex and real cases AMUB.

\section{Conclusion}
\label{sec:con}
Here we have analyzed general characteristics of AMUB constructed using combinatorial design techniques, using objects called Resolvable Block Design. We have identified parameters which critically influence the quality of AMUB. We have also shown that, large sets of real and complex class of Approximate MUBs, which we call $\beta$-AMUB can be constructed, in composite dimensions $(d = k\times s)$ with $|s-k| < \sqrt{d}$ using RBD.   In general in composite dimension, where only a very small set of MUBs is known, even in the complex case. 
 
In this work, we derive how $\beta$ depends on the nature of RBD$(X, A)$, with $|X|= d$. To do this, we broadly categorize RBD$(X, A)$ into two categories, one where all the parallel classes have a constant block size and another where they do not have a constant block size. When parallel classes have a constant size, say $k$, in  such RBD$(X, A)$ a single Hadamard matrix of order $k$ can be used to yield AMUBs over $\mathbb{C}^d$ (or $\mathbb{R}^d$), depending on whether complex or real Hadamard matrices are used. Here, for cases where parallel classes do not have a constant block size, one needs to use Hadamard matrices of the order of block sizes.

On the basis of our analysis and construction for various setting we conclude that $|s-k|$, $\mu$ and set of Block sizes $K$ are most critical parameters determining the closeness of AMUB to the MUBs constructed over $\mathbb{C}^d$ (or $\mathbb{R}^d$). When block sizes are near $\sqrt{d}$ and $\mu$ is $1$, we get $\mathcal{O}(\sqrt{d})$ many AMUBs for all $d$ with $2\delta = |s-k|< \sqrt{d}$, the $\beta = 1+ \frac{\delta}{\sqrt{d}} + \mathcal{O}(\frac{\delta^2}{d}) \leq 2$ and $\epsilon = 1- \mathcal{O}(d^{-\frac{1}{2}})$. And whenever a real Hadamard matrix of order $k$ or $s$ are available, we get $N(s)+1$ or $N(k)+1$ many ARMUBs with similar characteristics. 

Our experience suggest that constructing RBD$(X, A)$, having a large number of parallel classes, and having constant block size for all the parallel classes such block sizes remain near $\mathcal{O}(\sqrt{d})$ and the intersection number $\mu$ remains small  are generally difficult to achieve. In general our RBD$(X,A)$ has $\mathcal{O}(\sqrt{d})$ many parallel classes having $\mu = 1$ or $2$. We intend to work on constructing RBD$(X,A)$'s having larger order of parallel classes, keeping $\mu$ small and block sizes  near $\mathcal{O}(\sqrt{d})$, as our results show that such RBD$(X,A)$ will enable one to construct even larger number of AMUBs, having similar quality as in this work. Further we also intend to work on dispensing with the condition $|s-k|< \sqrt{d}$ and making it applicable for all the $d$'s. Even though our experience suggest that this condition is not very restrictive, and ratio of integers which satisfy this condition to the total numbers less than certain finite large integer is almost one. Nevertheless there are infinitely many integers which does not satisfy this condition, hence effort in the direction to dispense with this condition may be worthwhile.

\end{document}